\documentclass[prb,twocolumn,longbibliography,superscriptaddress,preprintnumbers]{revtex4-2}

\usepackage[colorlinks,bookmarks=true,citecolor=blue,linkcolor=blue,urlcolor=blue] {hyperref}
\usepackage{dcolumn,graphicx,amsfonts,amsthm,bm,color,appendix,float}
\usepackage{braket}
\usepackage[normalem]{ulem}
\usepackage[version=4]{mhchem}
\usepackage{siunitx}
\usepackage{amsmath}
\usepackage{amssymb}
\usepackage{CJKutf8}
\usepackage{bbm}
\usepackage{subfigure}
\usepackage{color}
\usepackage[mathlines]{lineno}

\newtheorem{theorem}{Theorem}

\begin{document}
\title{Characterization of higher-order topological superconductors using Bott indices}

\author{Xun-Jiang Luo}
\thanks{These authors contributed equally.}
\affiliation{School of Physics and Technology, Wuhan University, Wuhan 430072, China}
\author{Jia-Zheng Li}
\thanks{These authors contributed equally.}
\affiliation{School of Physics and Technology, Wuhan University, Wuhan 430072, China}
\author{Meng Xiao}
\email{phmxiao@whu.edu.cn}
\affiliation{School of Physics and Technology, Wuhan University, Wuhan 430072, China}
\affiliation{Wuhan Institute of Quantum Technology, Wuhan 430206, China}
\author{Fengcheng Wu}
\email{wufcheng@whu.edu.cn}
\affiliation{School of Physics and Technology, Wuhan University, Wuhan 430072, China}
\affiliation{Wuhan Institute of Quantum Technology, Wuhan 430206, China}

\begin{abstract}
The abundance of bulk and boundary topologies in higher-order topological phases offer remarkable tunability and diversity to boundary states but also pose a challenge to their unified topological characterization. In this work, we propose a theoretical framework to characterize time-reversal invariant topological superconductors hosting Majorana Kramers pairs (MKP) of corner states by using a series of spin Bott indices, which capture both bulk and boundary states topology. The developed invariants can characterize MKP in arbitrarily shaped systems and all distinct spatial distribution patterns of MKP. As an illustrative example, we apply our theory to analyze the Kane-Mele model with sublattice-dependent superconducting pairing potentials. In this representative model, both 
intrinsic and extrinsic higher-order topological superconductors can be realized and various patterns of MKP can be engineered through edge cleavage. Despite their high sensitivity to boundary terminations, MKP can be faithfully characterized by the proposed topological invariants. We further demonstrate the characterization of higher-order topological superconductors in the BDI symmetry class using Bott indices without resolving the spin degree of freedom.

\end{abstract}
\maketitle

\section{Introduction}
Higher-order topological phases, hosting gapless boundary states with codimension greater than one, have attracted extensive research interest in both theory \cite{Benalcazar2017a,Benalcazar2017,Khalaf2018,Geier2018,Zhang2019,Benalcazar2019,Wang2019,Trifunovic2019,Xu2019,Zhang2020b,Luo2021a,Luo2023,Lin2024,zhang2024topological,2024BTIluo,luo2025} and experiments \cite{20Li,Huang21s,Shumiya2022,WANG2022788,Xu23s,Hossain2024a,24Shafayat}. Higher-order topological superconductors (HOTSCs) featuring Majorana corner states \cite{Wang2018a,Zhu2019Second,Volpez2019,20zhang,2024Huang} are of particular interest because of their potential applications in topological quantum computation \cite{Nayak2008,Zhang2020c,Zhang2020d,Pan2022}. 
Various systems have been proposed to realize HOTSCs, such as superconducting-proximitized topological insulators \cite{Yan2018,Wang2018,Pan2019,Wu2020a,Tan2022}, odd-parity superconductors \cite{Yan2019,Ahn2020,Hsu2020,Huang2021}, and iron-based superconductors \cite{Zhang2019a,Wu2019,Wu2020,Chen2021a,Zhang2021,22Qin}.
In contrast to first-order topology, the topological boundary states of HOTSCs typically depend on not only  bulk states topology but also the boundary terminations \cite{Geier2018}.
The abundant bulk and 
 boundary topologies of HOTSCs endow boundary states with high tunability and richness but also
present the challenge of their unified topological characterizations \cite{19Benalcazar,22chen,Benalcazar2022,Luo2022,zhu2023,24Jahin,24Lin,2024zhuxiaoyu,zhu2024,Jiawei2024}. Several theoretical frameworks for characterizing higher-order topology have been developed, such as multipole moment \cite{Kang2019,Wheeler2019,Ono2019}, chiral multipole number \cite{Benalcazar2022}, and symmetry indicators \cite{Kruthoff2017,Po2017,Khalaf2018a,Skurativska2020}.
Despite these seminal works, a unified topological characterization of Majorana corner states remains elusive. A key difficulty is that the previously proposed invariants are defined by only bulk states
while Majorana corner states are sensitive to both bulk and boundary topologies \cite{Geier2018,22zhu}, which hinders establishing an exact correspondence between the proposed invariants and Majorana corner states. As such, a comprehensive topological characterization of HOTSCs must encompass both bulk and boundary state topologies, a challenge that existing frameworks have yet to fully address. Moreover, no current theory provides a complete characterization of HOTSCs with arbitrary shapes or captures the diverse real-space patterns of Majorana corner states.

Bott index, deeply rooted in $K$-theoretic invariants \cite{EXEL1991364,Hastings2010,LORING2015383,Loring2019AGT}, stands as a potent instrument and has been used to characterize various topological phases of matter \cite{Loring_2010,Titum2015,Huang181,YangYanBin2019,WangXS2020,LinLing2021,WangCitian2022}. These include topological insulators characterized by winding number \cite{EXEL1991364,LinLing2021}, Chern insulators \cite{Loring_2010,BottIndexTwo2021}, and quantum spin Hall insulators \cite{Huang181}. As a real space topological invariant, Bott index is defined by a couple of unitary matrices. When the unitary matrix is generated by the real space polynomial $e^{i2\pi x/L}$ and/or $e^{i2\pi y/L}$, the Bott indices can provide equivalent expressions for winding numbers or Chern numbers \cite{LinLing2021,BottIndexTwo2021}. Recently, the Bott index was used to characterize chiral symmetric higher-order topological insulators by generalizing the form of real space polynomial $e^{i2\pi f(\bm r)}$ \cite{Benalcazar2022,Jiazheng2024}. This break through motivates us to explore the general topological characterization of HOTSCs by Bott index.

 \begin{figure*}
\centering
\includegraphics[width=1\textwidth]{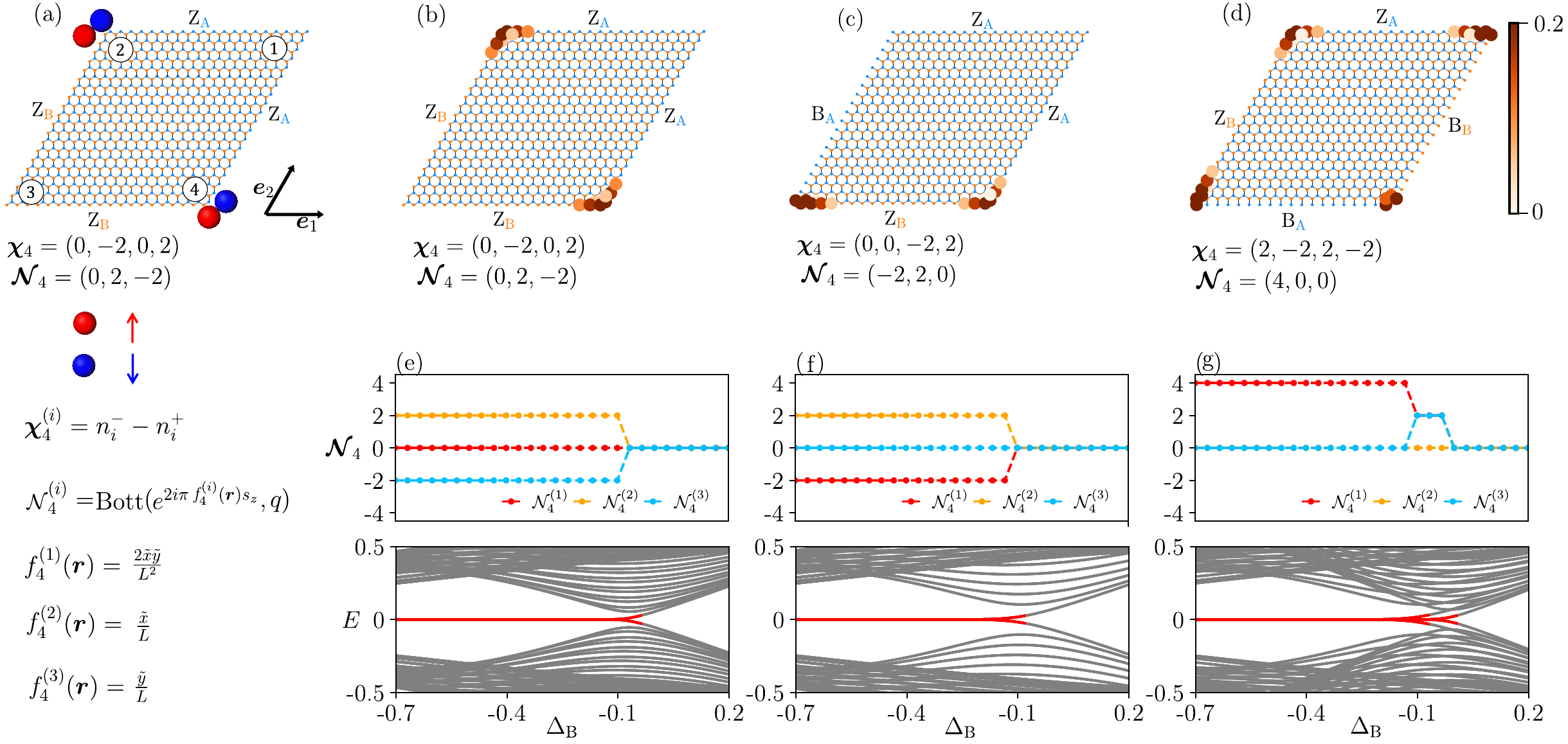}
\caption{(a) The schematic illustration of  MKP and spin Bott index. $\text{Z}_{\alpha}$ ($\text{B}_{\alpha}$) labels the zigzag (bearded) edge formed by the atom $\alpha \in \{A, B\}$. The circled numbers at the corners label the order of corners. $n_i^{+}$ ($n_i^{-}$) denotes the number of corner staes at the $i$th corner with eigenvalue $+$ ($-$) of $\mathcal{C}$. (b)-(d) The different distribution patterns of MKP in diamond-shaped systems with different edge terminations. The colorbar denotes the local density of states of MKP. (e)-(g) The evolution of  ${\mathcal{N}}_4^{(1,2,3)}$ as functions of $\Delta_{\text{B}}$ to characterize the patterns of MKP in systems with geometries shown in (b)-(d). The lower panels of (e)-(g) plot the energy spectrum, where the red bands mark MKP. The common parameters in (b)-(g) are taken as $t=1,\lambda_{\text{so}}=0.2,\Delta_{\text{A}}=0.5,\mu=0,\lambda_{v}=\lambda_{\text{R}}=0$. In (b)-(d), $\Delta_{\text{B}}=-\Delta_{\text{A}}$. In the calculation of (e)-(f), the side length of the system is $L=40$ in unit of lattice constant.}
\label{N4}
\end{figure*}

In this work, we develop a theoretical framework to characterize time-reversal invariant HOTSCs by higher-order topological invariants based on a series of spin Bott indices.
Our approach is motivated by recent theoretical progresses of characterization \textit{chiral} symmetric systems by using Bott indices \cite{Benalcazar2022,Jiazheng2024}. 
Time-reversal invariant superconductors are in the DIII symmetry class and naturally host chiral symmetry $C=-iTP$, where $T$ and $P$ are, respectively, time-reversal and particle-hole symmetry. However, a direct application of the Bott index $N$ proposed in Ref.~\onlinecite{Jiazheng2024} based on the chiral symmetry $C$ of superconductors can not capture the higher-order topology since $N$ vanishes as required by the  $T$ symmetry (see Eqs.~\eqref{eq2} and \eqref{eql}). 
Therefore, we develop spin Bott indices to characterize HOTSCs, in analogy to the spin-resolved topological characterization of quantum spin Hall insulators \cite{18huang}.
We emphasize that the Bott indices here are obtained under open boundary conditions, which enables capturing the topology of both bulk and boundary states. 
Particularly, the defined spin Bott indices can characterize MKP in systems with arbitrary shapes and all distinct spatial distribution patterns of MKP. 
To demonstrate the effectiveness of our theory, we study the Kane-Mele model \cite{05kane} with sublattice-dependent superconducting pairing potentials. In this model, MKP can be flexibly engineered by boundary cleavage, as shown in Figs.~\ref{N4} and \ref{N6},  and all different patterns of MKP can be characterized by the proposed invariants. In addition, we also investigate disorder effects on the HOTSCs informed by the invariants. Furthermore, by using Bott indices without resolving the spin degree of freedom, we  demonstrate the characterization of HOTSCs in the BDI symmetry class.

This paper is organized as follows. In Sec.~\ref{II}, we develop the spin Bott index for the systems of DIII symmetry class.  In Sec.~\ref{III}, we develop the framework for characterizing MKP by spin Bott indices for systems with or without spin conservation. In Sec.~\ref{IV}, we use the Kane-Mele model with the sublattice-dependent superconducting pairing potentials to examine our theory.  In Sec.~\ref{V}, we further demonstrate the topological characterization of HOTSCs in the BDI symmetry class by Bott indices.
In Sec.~\ref{VI}, we present a brief discussion and summary. Appendices \ref{appendixa}-\ref{appendixg} complement the main text with additional technical details.


\begin{figure*}
\centering
\includegraphics[width=0.98\textwidth]{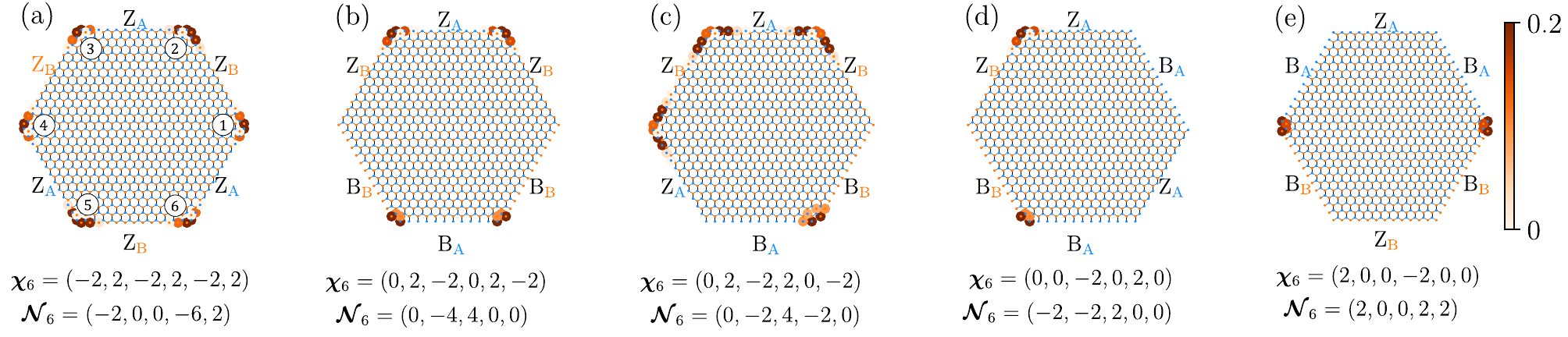}
\caption{(a)-(e) The distribution patterns of MKP in hexagon-shaped systems with a few typical edge terminations.  The vectors $\bm{\chi}_6$ and $\bm{\mathcal{N}}_6$ describe and characterize the different patterns of MKP, respectively.  The  parameters  are taken as $t=1,\lambda_{\text{so}}=0.2,\Delta_{\text{A}}=-\Delta_{\text{B}}=0.5,\mu=0,\lambda_{v}=\lambda_{\text{R}}=0$.}
\label{N6}
\end{figure*}

\section{Spin Bott indices}
\label{II}
The symmetry operators of superconductors in the DIII symmetry class can be represented as $T=is_y\mathcal{K}$, $P=\tau_ys_y\mathcal{K}$, and $C=\tau_y$ with Pauli matrices $s_y$ and $\tau_y$  acting, respectively, on the spin and particle-hole spaces, and $\mathcal{K}$ being the complex conjugation operator. We first show that the application of the previously defined Bott indices \cite{Benalcazar2022,Jiazheng2024} to superconductors in the DIII class based on $C$ can not capture the higher-order topology.

In the eigenbasis of $C$ where  $C$ is represented by $\tau_z$,
the Bogoliubov–de Gennes (BdG) Hamiltonian can be expressed in an off-diagonal form, 
\begin{eqnarray}
H=\begin{pmatrix} 0 & h\\
h^{\dagger}& 0 
\end{pmatrix}.
\label{eqh}
\end{eqnarray}
Using singular value decomposition $h=U_A\Sigma U_B^{\dagger}$, the Bott index can be defined in the following equivalent ways (see Appendix \ref{appendixa} for details) 
\begin{eqnarray}
N&=&\text{Bott}(m,q)\equiv\frac{1}{2\pi i}\text{Tr}\text{log}(mqm^{\dagger}q^{\dagger})\nonumber\\
\quad\quad\quad&=&\frac{1}{4\pi i}\text{Tr}[C\text{log}({M}Q{M}^{\dagger}Q)]\nonumber\\
&=&\frac{1}{4\pi i}\text{Trlog}(\Xi),
\label{eq2}
\end{eqnarray}
where $m=e^{2\pi if(\bm r)}$ is a unitary matrix with $f(\bm r)$ being a polynomial  of the position operator $\bm r$, $q=U_AU_B^{\dagger}$, $Q$ is obtained by replacing $h$ in Eq.~\eqref{eqh} with $q$,  ${M}=\tau_0\otimes m$, and $\Xi=\frac{\mathbbm{1}+C}{2}{M}Q{M}^{\dagger}Q+\frac{\mathbbm{1}-C}{2}{M}^{\dagger}Q{M}Q$.  
Because $TCT^{-1}=-C$, $TMT^{-1}=M^{\dagger}$, and $TQT^{-1}=Q$, we have
\begin{eqnarray}
T\Xi T^{-1}=\Xi,
\label{eql}
\end{eqnarray}
which implies that the eigenvalues of matrix $\Xi$ come into the time-reversal pairs $(e^{i\lambda_i},e^{-i\lambda_i})$. Therefore, $N$ is exactly zero as constrained by the $T$ symmetry.

This restriction is consistent with the MKP related by $T$ having opposite eigenvalues of $C$. We denote the  MKP at a given corner as $\gamma_{1}$ and $\gamma_{2}$, with $\gamma_{2}=T\gamma_{1}$. The Majorana corner states can always be chosen as the eigenstates of $C$ through a superposition of $\gamma_{1}$ and $\gamma_{2}$. Without loss of generality, we choose $C\gamma_{1}=\gamma_{1}$, which leads to 
\begin{eqnarray}
C\gamma_{2}=CT\gamma_{1}=-T\gamma_{1}=-\gamma_{2}.
\label{eq4}
\end{eqnarray}
Therefore, a successful characterization should resolve the spin degree of freedom, which motivates us to construct spin Bott indices.

We first focus on the case with spin U(1) symmetry, for example, $[s_z, H]=0$. Then, there is another chiral symmetry $\mathcal{C}=s_zC$, which satisfies $\{\mathcal{C},H\}=0$ and $T\mathcal{C}T^{-1}=\mathcal{C}$. Following Eq.~\eqref{eq4}, the MKP at a given corner have the same eigenvalue of $\mathcal{C}$. We define the spin Bott index $\mathcal{N}$ by utilizing the $\mathcal{C}$ operator,
\begin{eqnarray}
\label{eqe1}
\mathcal{N}&=&\frac{1}{4\pi i}\text{Tr}[\mathcal{C}\text{log}({M}Q{M}^{\dagger}Q)].
\end{eqnarray}
The spin Bott index can also be defined by the  $C$ operator but with the matrix $M$ replaced by $M_z=e^{2\pi i f(r)s_z\tau_0}$,
\begin{eqnarray}
\bar{\mathcal{N}}&=&\frac{1}{4\pi i}\text{Tr}[C\text{log}({M}_zQ{M}_z^{\dagger}Q)].
\label{eqn}
\end{eqnarray}
Here $M_z$
commutes with $T$. With $s_z$ conservation,  $\mathcal{N}=\bar{\mathcal{N}}$ (see Appendix \ref{appendixb} for details).

\section{Characterization of  MKP}
\label{III}
When $[s_z, H]=0$, $H$ is block-diagonal in the spin space and each block belongs to the AIII symmetry class \cite{Chiu2016}, allowing one corner to have $Z$ pairs of MKP. To describe the pattern of MKP at corners of a system with open-boundary conditions,
we define the vector
\begin{eqnarray}
\boldsymbol{\chi}_{p}=\left(n_1^{-}-n_1^{+},\dots,n_p^{-}-n_p^{+}\right) \in 2\bm{Z},
\label{chiv}
\end{eqnarray}
where $n_i^{\pm}$ denote the number of Majorana corner states  localized at the $i$th corner with eigenvalue $\pm 1$ of $\mathcal{C}$ and $p$ is the number of corners.

Building upon our joint work \cite{Jiazheng2024}, for arbitrarily shaped systems,
we characterize $\bm{\chi}_p$ using ($p-1$) distinct spin Bott indices,
\begin{eqnarray}
\boldsymbol{\chi}_{p}=\mathcal{M}^{-1}\cdot\left({\mathcal{N}}_p^{(1)},\dots,{\mathcal{N}}_p^{(p-1)},0\right)^{\mathrm{T}},
\label{cc1}
\end{eqnarray}
where ${\mathcal{N}}_p^{(i)}\in 2Z$ is generated by  a polynomial ${f}_p^{(i)}$ of position operator for $1\le i \le p-1$. The matrix $\mathcal{M}$ is defined by $\mathcal{M}_{ij}=\operatorname{sign}\left({f}_p^{(i)}(\boldsymbol{x}_{j})\right)/2$ and $\mathcal{M}_{pj}=1/2$, with $\bm{x}_j$ being the position of $j$th corner.  The polynomials are constructed such that $\text{det}(\mathcal{M})\neq 0$ and ${f}_p^{(i)} (\bm{x}_j)=\pm 1/2$ (coordinate origin chosen at the center of the system) \cite{Jiazheng2024}. We emphasize that here  ${\mathcal{N}}_p^{(i)}$ is obtained under the open boundary conditions. In this case, although $q$ is not unique when zero-energy states are present, ${\mathcal{N}}_p^{(i)}$ is still well-defined \cite{Jiazheng2024}.

For example, 
${f}_{p}^{(i)}$ for the diamond- ($p=4$ in Fig.~\ref{N4}) and hexagon-shaped ($p=6$ in Fig.~\ref{N6}) systems can be chosen, respectively, as
\begin{eqnarray}
&&{f}_{4}^{(1)}=2\tilde{x}\tilde{y}/L^2,\quad {f}_{4}^{(2)}=\tilde{x}/L,\quad {f}_{4}^{(3)}=\tilde{y}/L,\nonumber\\
&&{f}_{6}^{(1)}=(x^3-xy^2/3+8\sqrt{3}y^3/9)/2L^3,\nonumber\\
&&{f}_{6}^{(2)}=(x^2-4xy/\sqrt{3}-y^2/3)/2L^2,\nonumber\\
&&{f}_{6}^{(3)}=(x^2+4xy/\sqrt{3}-y^2/3)/2L^2,\nonumber\\
&&{f}_{6}^{(4)}=(x^3-3xy^2)/2L^3,\nonumber\\
&&{f}_{6}^{(5)}=(x^3+7xy^2/3)/2L^3,
\label{fe}
\end{eqnarray}
where $L$ is the length of the side and $(x,y)$ is the coordinate of lattice sites. For $p=4$, $(\tilde{x},\tilde{y})$ is defined through $(x,y)=\tilde{x}\bm{e}_1+\tilde{y}\bm{e}_2$ with  $\bm{e}_{1,2}$ being unit vectors along sides of the diamond [Fig.~\ref{N4}(a)]. It is noted that the choice of polynomials $f_p^{(i)}(\bm r)$ is not unique for certain shaped system. In Appendix \ref{appendixp}, we present an alternative choice of $f_6^{(i)}(\bm r)$ with $i=1,2,3,4,5$.

We now turn to the case without spin U(1) symmetry, where the operator $\mathcal{C}$ is no longer the chiral symmetry of $H$. Therefore, $\mathcal{N}$ in Eq.~\eqref{eqe1} and $\bm{\chi}_{p}$ in Eq.~\eqref{chiv} are no longer well-defined. However, $\bar{\mathcal{N}}$ is still well-defined if the spin mixing term is not too strong, of which the exact mathematical condition is derived in Appendix \ref{appendixc}. Without the spin U(1) symmetry, MKP has a $Z_2$ topological classification. We redefine the correspondence between the spin Bott indices and the pattern of MKP as
\begin{eqnarray}
\bar{\boldsymbol{\chi}}_{p}= \Big[\mathcal{M}^{-1}\cdot\left({\bar{\mathcal{N}}}_p^{(1)},\dots,{\bar{\mathcal{N}}}_p^{(p-1)},0\right)^{\mathrm{T}}\Big] ~\text{mod}~ 4.
\label{cc2}
\end{eqnarray}
Here ${\bar{\mathcal{N}}}_p^{(i)}$ is still generated by the same $ f_{p}^{(i)}$ through Eq.~\eqref{eqn} and each element of $\bar{\boldsymbol{\chi}}_{p}$ is 0 or 2, which counts the number of topologically robust Majorana corner states at a given corner.

\section{Theoretical model}
\label{IV}
As an illustration, we study the Kane-Mele model with sublattice-dependent superconducting pairing potentials. The model Hamiltonian is
\begin{eqnarray}
\mathcal{H}&=&t \sum_{\langle i j\rangle,s } c_{is}^{\dagger} c_{js}+i\lambda_{\text{so}} \sum_{\langle\langle i j\rangle\rangle,s,s^{\prime}} \nu_{i j} c_{is}^{\dagger}(s_z)_{ss^{\prime}}c_{js^{\prime}}\nonumber\\
&+&\sum_{i,s}(\lambda_{{v}}\xi_i-\mu) c_{is}^{\dagger}c_{is} \nonumber+i\lambda_{\text{R}}\sum_{\langle i j\rangle,s,s^{\prime}}c_{is}^{\dagger}\left((\bm s \times \hat{\bm d}_{i j})_z\right)_{ss^{\prime}} c_{js^{\prime}}\nonumber\\
&+&
\sum_{i,s,s^{\prime}} (\Delta_i c_{is}^{\dagger}(-is_y)_{ss^{\prime}} c_{is^{\prime}}^{\dagger}+h.c.) ,
\label{Ha}
\end{eqnarray}
where $c_{is}^{\dagger}$ ($c_{is}$) is the electron creation  (annihilation) operator with spin index $s$, and $\bm{s}=(s_x,s_y,s_z)$ are spin Pauli matrices. The first (second) term in $\mathcal{H}$ describes the hopping between the nearest- (next-nearest-) neighbors on honeycomb lattice and $\nu_{i j}=(2 / \sqrt{3})(\hat{\bm d}_1 \times \hat{\bm d}_2)_z= \pm 1$, where $\hat{\bm d}_{1,2}$ are unit vectors along the two bonds which the electron traverses from site $j$ to $i$. The third term describes the staggered potential with $\xi_{i}=1$ ($-1$) for the A (B) sublattice and $\mu$ is the chemical potential.
The fourth term is the nearest-neighbor Rashba term. The last term describes the sublattice-dependent superconducting pairing potentials with $\Delta_{i}=\Delta_{\text{A}}$ ($\Delta_{\text{B}}$) for the A (B) sublattice. 
Here,  $t,\lambda_{\text{so}},\lambda_{{v}},\lambda_{\text{R}},\Delta_{\text{A}}$,  and $\Delta_{\text{B}}$ are model parameters and we take $\lambda_{v}=\lambda_{\text{R}}=0$ unless otherwise stated. With $\lambda_{\text{R}}=0$, there is a spin U(1) symmetry of the BdG Hamiltonian $H$, $[H,s_z]=0$, in the Nambu basis defined by $\bm{\Psi}=\{\bm{\psi},is_y\bm{\psi}^{\dagger}\}$ with $\bm{\psi}$ being the basis of the normal states.
We note that this model 
can be realized in quantum spin Hall insulators on a buckled honeycomb lattice  \cite{Liucc2011,Liu2011a,Xu2013,Si2014,Luo2024}. In these systems, the sublattice-dependent superconducting pairing potentials could be obtained by covering superconductors on top and bottom surfaces with distinct pairing potentials.

The normal state of $\mathcal{H}$ realizes quantum spin Hall insulator when the chemical potential $\mu$ is in the bulk gap and hosts gapless helical states along edges. 
For simplicity, we take $\mu=0$ to analyze the edge states  (see Appendix \ref{appendixd} for details).  
The helical edge states along
the zigzag and bearded edges (type-I), formed by either the A or B atoms, are gapped by the superconducting pairing with magnitude $\Delta_{\text{A}}$ or $\Delta_{\text{B}}$ (see Appendix \ref{appendixd} for details). In comparison, the helical edge states along the armchair edge (type-II), formed by both the A and B atoms, are gapped by the superconducting potentials with magnitude $(\Delta_{\text{A}}+\Delta_{\text{B}})/2$. MKP emerge at a given corner when the pairing gaps of two adjacent edges have a sign change.  This scenario can occur in two special cases, (i) both adjacent edges belong to type-I edge and are formed by A and B atoms, respectively, with $\Delta_{\text{A}}\Delta_{\text{B}}<0$; (ii) two adjacent edges belong to type-I and type-II edges, respectively, satisfying $(\Delta_{\text{A}}+\Delta_{\text{B}})\Delta_{\alpha}<0$ (type-I edge is formed by the $\alpha$ atom). 
For case (i) [(ii)], we find that MKP can be engineered in systems with diamond and hexagon shapes (for square and dodecagon shapes). In the following,   we focus on the case (i),  and case (ii) is discussed in Appendix \ref{appendixe}.

\begin{figure}
\centering
\includegraphics[width=0.5\textwidth]{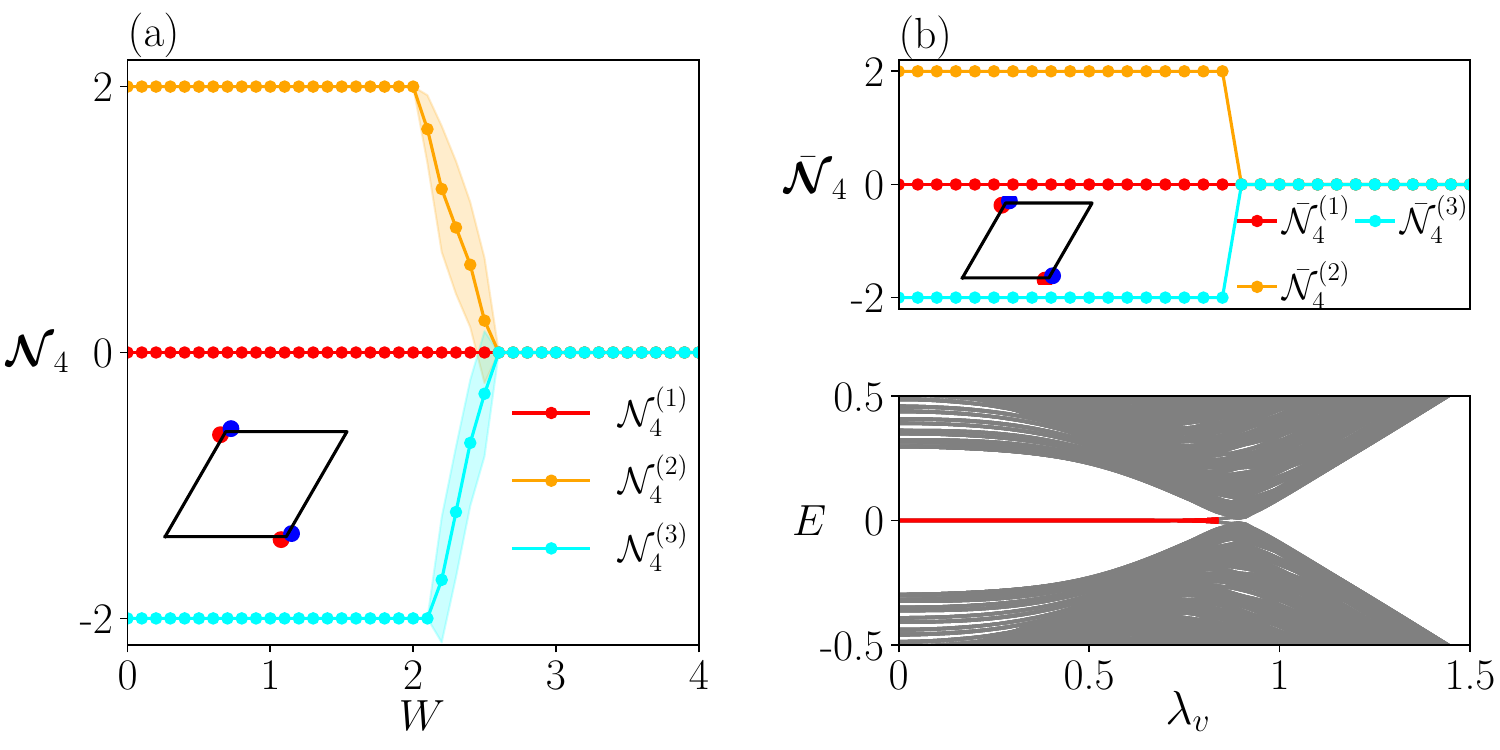}
\caption{
(a) The Bott indices ${\mathcal{N}}_4^{(1,2,3)}$ as functions of disorder strength $W$ averaged over 200 independent disorder realizations. The solid line and shaded region indicate the average
of the plotted quantity and the standard deviation, respectively. 
(b) The invariants $\bar{{\mathcal{N}}}_4^{(1,2,3)}$ and the energies $E$ as functions of $\lambda_v$ by taking $\lambda_{\text{R}}=0.25$. In (a) and (b), the system geometry is shown in Fig.~\ref{N4}(b), and the side length $L=30$. Other unspecified model parameters are the same as used in Fig.~\ref{N4}(b). The insets schematically plot the pattern of MKP.}
\label{dis}
\end{figure}

In the diamond- [Fig.~\ref{N4}] and hexagon-shaped [Fig.~\ref{N6}] systems, all edges belong to type-I and are formed by A or B atom, depending on the specific boundary termination. Therefore,  all possible distribution patterns of MKP can be realized by varying edge cleavage to design sign patterns of edge-state energy gap. Specifically, the number of corners that host MKP can be any even number not exceeding 4 (6) for the diamond (hexagon) system. In Figs.~\ref{N4}(b)-\ref{N4}(d) and Figs.~\ref{N6}(a)-\ref{N6}(e), we present several typical spatial distribution patterns of MKP,  demonstrating the flexibility to engineer MKP.
We employ the vectors $\bm{\chi}_{4}$ and $\bm{\chi}_{6}$ to describe these configurations of MKP, which can be fully characterized 
by the vectors $\bm{\mathcal{N}}_4=({\mathcal{N}}_{4}^{(1)},{\mathcal{N}}_{4}^{(2)},{\mathcal{N}}_{4}^{(3)})$ and $\bm{\mathcal{N}}_6=({\mathcal{N}}_6^{(1)}, {\mathcal{N}}_6^{(2)}, {\mathcal{N}}_6^{(3)}, {\mathcal{N}}_6^{(4)},{\mathcal{N}}_6^{(5)})$, respectively.  Our numerical results are consistent with Eq.~\eqref{cc1}. In Figs.~\ref{N4}(e)-\ref{N4}(g), we present the evolution of $\bm{\mathcal{N}}_{4}$ and energies as functions of $\Delta_{\text{B}}$ for the systems  depicted in Figs.~\ref{N4}(b)-\ref{N4}(d), which consistently characterizes the topological phase transitions. We emphasize that this correspondence captured by $\bm{\mathcal{N}}_{4,6}$  reflects the nontrivial topology of both bulk and edge states \cite{Jiazheng2024}. In contrast, the Bott index defined under the periodic boundary conditions \cite{Benalcazar2022} can not precisely characterize MKP as it does not capture termination-dependent edge topology.

The real space topological invariants are particularly useful to study the disorder effect on topological states \cite{Li2020a,Benalcazar2022}.  MKP are protected by the $P$ and $T$ symmetries and are robust against weak disorders. To show this, we add the disorder term $\sum_i\mu_ic_i^{\dagger}c_i$ to Eq.~\eqref{Ha}, with $\mu_i$ obeying uniform random distribution in $[-W, W]$.  In Fig.~\ref{dis} (a),  we study the disorder effect on the system shown in Fig.~\ref{N4}(b).  We find that $\bm{\mathcal{N}}_4$ remains quantized at (0, 2, -2) for weak disorders, indicating the persistence of the MKP pattern, until a transition
drives the system into a trivial phase with $\bm{\mathcal{N}}_4=(0,0,0)$ when
the disorder becomes sufficiently strong.

We further consider the $s_z$ symmetry breaking case with $\lambda_{\text{R}}\neq 0$. In this case,  we 
characterize MKP by $\bar{\mathcal{N}}$ defined in Eq.~\eqref{eqn}. For instance, we consider nonzero $\lambda_{\text{R}}$ in the system associated with Fig.~\ref{N4}(b). In Fig.~\ref{dis}(b),  we present the numerical results of $\bar{{\mathcal{N}}}_4^{(1,2,3)}$ and energies as functions of $\lambda_{v}$ at a fixed $\lambda_{\text{R}}$. The topological invariants correctly characterize the topological phase transition driven by the increasing of $\lambda_{v}$.

\begin{figure}
\centering
\includegraphics[width=0.5\textwidth]{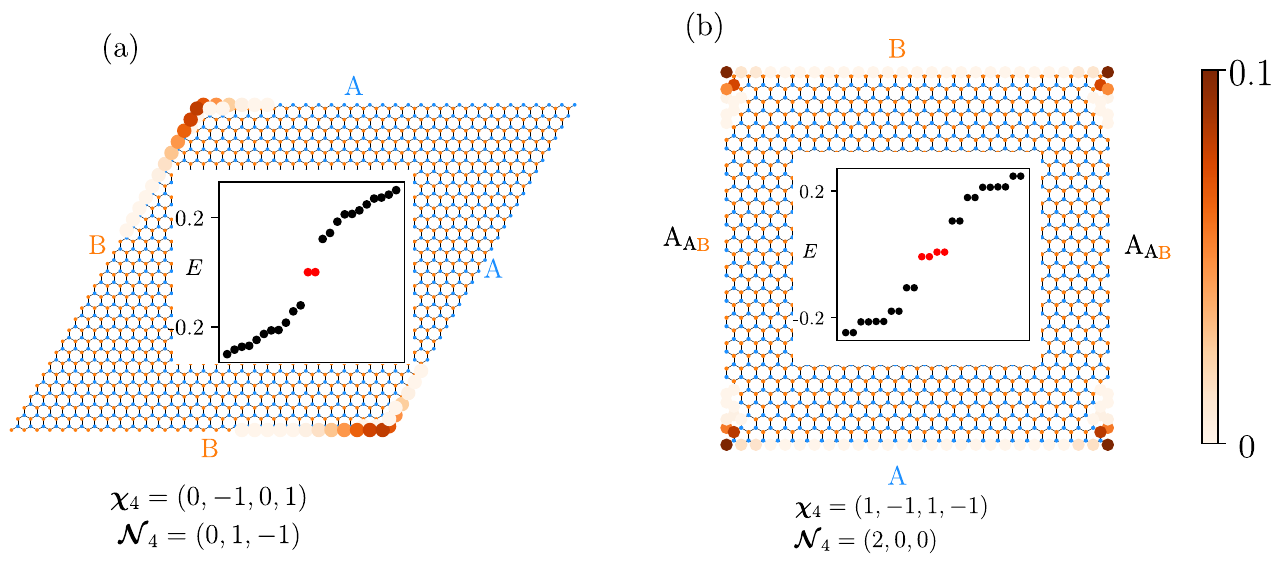}
\caption{Majorana corner states in different-shaped systems that belong to the BDI symmetry class. The letters $\text{A}_{\text{AB}}$,  $\text{Z}_{\alpha}$,  $\text{B}_{\alpha}$ label the armchair, zigzag, and bearded edges, respectively, with $\alpha\in$ \{A, B\}.  The insets plot the energies close to zero. In (a), we take $\Delta=0.25,\eta_{\text{B}}=0.4,\eta_{\text{A}}=0$. In (b), we take $\Delta=0.25,\eta_{\text{B}}=\eta_{\text{A}}=0.4$.}
\label{BDI}
\end{figure}

\section{HOTSC in the BDI symmetry class}
\label{V}
We now turn to another important superconducting system in the BDI symmetry class, which hosts the chiral symmetry as a combination of effective time-reversal symmetry $\tilde{T}$ and particle-hole symmetry $P$ with $\tilde{T}^2=1$ and $P^2=1$. 
For the BDI symmety class, the Majorana corner states are protected by chiral symmetry and have a $Z$ topological classification. Therefore, the Majorana corner states can be directly characterized by the Bott index $N$ defined by Eq.~\eqref{eq2} without the need of resolving the spin degree of freedom. The different Majorana  patterns in the real space space described by $\bm{\chi}_{p}$ ($\bm{\chi}_{p}\in \bm{Z}$) can be  characterized by $\bm{N}_p$ ($\bm{N}_{p}\in \bm{Z}$) as
\begin{eqnarray}
\boldsymbol{\chi}_{p}=\mathcal{M}^{-1}\cdot\left({\bm{N}}_p,0\right)^{\mathrm{T}}.
\label{cc12}
\end{eqnarray}
Here $\bm{N}_{p}=({{N}}_p^{(1)},\dots,{{N}}_p^{(p-1)})$
and ${N}_p^{(i)}$ is generated by  the polynomial ${f}_p^{(i)}$ of position operator for $1\le i \le p-1$.

We note that the BDI symmetry class can be realized in two special cases by applying a magnetic Zeeman field or considering magnetic orders on top of the DIII symmetry class systems which respect spin SU$(2)$ [case (a)] or U$(1)$ [case (b)] symmetry. For the case (a), the magnetic field can be applied along arbitrary $\bm n$ direction  and the $\tilde{T}$ symmetry is defined as $\tilde{T}=Ts_{\bm e}$ with $\{s_{\bm n},s_{\bm e}\}=0$. For the case (b), the magnetic field direction is constrained to be in the plane perpendicular to the spin U$(1)$ $z$-axis and $\tilde{T}$ symmetry is defined as $Ts_z$. The case (a) can occur in the superconducting systems with collinear magnetic structures described by spin point group \cite{Liuqihang2022}. Even for the coplanar magnetic structures ($xy$ plane assumed) systems, the BDI symmetry class can still be realized with the definition of $\tilde{T}=Ts_z$.
The case (b) can be exemplified  by our model (Eq.~\eqref{Ha}) by adding an in-plane magnetic Zeeman field and setting $\lambda_{\text{R}}=0$. In the following, we focus on the HOTSCs in the BDI symmetry class for this model.

 To engineer Majorana corner states, we consider site-dependent Zeeman field which can be obtained in the Kane-Mele-Hubbard model \cite{Luo2024}. The total model Hamiltonian in momentum space can be written as
\begin{equation}
\begin{aligned}
 {H}(\bm k)&=f_x(\bm k)\tau_z\sigma_xs_0+f_y(\bm{k})\tau_z\sigma_ys_0+f_z(\bm{k})\tau_z\sigma_zs_z\\
&+\Delta\tau_xs_0+\eta_{\text{A}}\tau_0(\sigma_0+\sigma_z)/2s_{x}+\eta_{\text{B}}\tau_0(\sigma_0-\sigma_z)/2s_{x},
\label{h}   
\end{aligned}
\end{equation}
 where $f_{x,y,z}$ are momentum-dependent functions (see Appendix \ref{appendixd}) and Pauli matrices $\sigma$ act on the orbitals. We assume an $s$-wave pairing with $\Delta_{\text{A}}=\Delta_{\text{B}}=\Delta$. $\eta_{\text{A}}$ and $\eta_{\text{B}}$ denote the magnitude of Zeeman field at the A and B sites, respectively. $H$ respects the $\tilde{T}$ and $P$ symmetries, where $\tilde{T}=Ts_z=-s_xK$ and $P=\tau_ys_yK$. The chiral symmetry can be represented by combination of $\tilde{T}$ and $P$ symmetries, $\mathcal{C}=i\tilde{T}P=\tau_ys_z$. Thus, $H$ in Eq.~\eqref{h} belongs to the BDI symmetry class. In the following, we show that various distinct spatial patterns of Majorana corner states can be realized  by tuning $\Delta, \eta_{\text{A}}$ and $\eta_{\text{B}}$.

 According to the edge theory analysis (see Appendix \ref{appendixd}),  the gapless edge states at the Dirac points along the armchair and zigzag directions are the approximate eigenstates of operators $\sigma_ys_z$ and $\sigma_z$, respectively. Thus, the edge states along the armchair direction can be gapped by the terms $\Delta\tau_xs_0, \eta_{\text{A}}\sigma_zs_x/2$, and $-\eta_{\text{B}}\sigma_zs_x/2$, which give rise to the edge states energy gap $\Delta-|\eta_{\text{A}}-\eta_{\text{B}}|/2$. The edge states along the zigzag (beard) direction formed by the $\alpha$ atom can be gapped by the terms $\Delta\tau_xs_0$ and $\eta_{{\alpha}}(\sigma_0+\xi\sigma_z)s_x/2$ with $\xi=1,-1$ for $\alpha=\text{A},\text{B}$, which give rise to the edge states energy gap $\eta_{\alpha}-\Delta$. When $\eta_{\text{A}}=\eta_{\text{B}}=\eta$, our analysis reproduces the results obtained in Ref.~\onlinecite{Pan2019}.

The Majorana corner states can be realized when the two
adjacent edges have a Zeeman-dominated and superconductivity-dominated gap, respectively. This scenario can be realized in two special cases, (i) both adjacent edges belong to type-I edge and are formed by A and B atoms, respectively, with $(|\eta_{\text{A}}|-|\Delta|)(|\eta_{\text{B}}|-|\Delta|)<0$; (ii) two adjacent edges belong to type-I and type-II edges, respectively, satisfying  $|\eta_{\alpha}|>|\Delta|>|\eta_{\text{A}}-\eta_{\text{B}}|/2$ (type-I edge is formed by the $\alpha$ atom). Thus, by tuning $\Delta, \eta_{\text{A}}$ and $\eta_{\text{B}}$, various distinct spatial patterns of the Majorana corner states can be realized in the diamond, square, hexagon, and dodecagon systems.  In Fig. ~\ref{BDI}, we present certain Majorana patterns in the diamond and square systems. These different Majorana patterns can be characterized by the corresponding $\bm{N}_4$ vectors, as listed in Fig. ~\ref{BDI}.

\section{Discussion and summary}
\label{VI}
We discuss some details of our theoretical approach. Although our numerical results are given for system with certain geometry, our theory is applicable for arbitrarily shaped system. For a general system with $p$ corners, we can construct $(p-1)$ polynomials $f_p^i(\boldsymbol{r})$ that satisfy the required conditions (see the work \cite{Jiazheng2024} for details).  Based on the exact correspondence between the generated Bott indices and patterns of Majorana corner modes, we can fully characterize higher-order topological superconductors.

We emphasize that crystalline symmetries can generally impose constraints on the Bott indices by restricting the patterns of Majorana corner modes that to be compatible with the symmetry. For example, in a time-reversal invariant system with four-fold rotation symmetry, MKP patterns that violate this symmetry are prohibited. Specifically, MKP must appear simultaneously at all four corners or not at all. Therefore, the constraints on the spin Bott index yield $\mathcal{N}_4^{(2)}=\mathcal{N}_4^{(3)}=0$ and $\mathcal{N}_4^{(1)}=4Z$ with $Z$ being an integer. This principle extends to other crystalline symmetries and constrains the possible values of the Bott indices. A detailed exploration of specific constraints for various crystalline symmetries is an interesting direction for future
work.

Higher-order topological phases have been classified into intrinsic type which hosts crystalline symmetry-protected bulk topology and extrinsic type whose topological phase transition can be tuned by closing boundary energy gap \cite{Geier2018,Trifunovic2019}. Both the intrinsic and extrinsic HOTSCs can be realized by $\mathcal{H}$ in Eq.~\eqref{Ha}. In the case of $\lambda_{\text{R}}=0$ and $\Delta_{\text{A}}=-\Delta_{\text{B}}$, $\mathcal{H}$ describes an odd parity superconductor with an effective inversion symmetry in the bulk, where a $Z_4$ inversion symmetry indicator \cite{Hsu2020} $\kappa$ can be defined as a bulk invariant to characterize the bulk topology (see Appendix \ref{appendixg} for details). The intrinsic HOTSCs are realized in the three examples shown in Figs.~\ref{N4}(b), ~\ref{N6}(a), and ~\ref{N6}(e), where both the bulk and edge terminations preserve the inversion symmetry. Here the nontrivial bulk invariant $\kappa$ indicates the presence of MKP corner states for symmetry-preserving edge terminations but does not distinguish the different spatial patterns of MKP in the above three examples.  In the case where the crystalline symmetry is broken in the bulk and/or by edge termination, extrinsic HOTSCs with time-reversal symmetry can also host robust MKP protected by (1) energy gaps of both bulk and boundaries and (2) symmetries of $P$ and $T$, as clearly demonstrated by results in Figs.~\ref{N4} and \ref{N6}. 
Our theory can characterize both intrinsic and extrinsic time-reversal invariant HOTSCs and distinguish different spatial patterns of MKP, while symmetry indicators such as $\kappa$ cannot.

In superconductors belonging to the D symmetry class, where only the particle-hole symmetry 
$P$ is present, Majorana corner states are protected by this symmetry. 
This scenario can be realized, for example, by further adding a Rashba term to Eq. \eqref{h}, where a small $\lambda_{\text{R}}$ cannot remove the Majorana corner states but can change the symmetry classes. The generalization of our theory to the D symmetry class is an open question.

In summary, we establish a real-space method for characterization of HOTSCs in both DIII and BDI symmetry classes using invariants of Bott indices. We apply our theory to characterize various patterns of Majorana corner states in a representative model with different shapes. Our study introduces Bott indices for characterizing higher-order topological superconductors, which advances the understanding of their  properties and opens up a range of applications.

\section{Acknowledgments}
X.-J. L. and F.W. are supported by the National Key Research and Development Program of China (Grant No.
2022YFA1402400) and the National Natural Science Foundation of China (Grant No. 12274333). 
J.-Z. L. and M.X. are supported by the National Key Research and Development Program of China (Grant No. 2022YFA1404900), the National Natural Science Foundation of China (Grant No. 12274332).

\appendix

\section{equivalent definitions of the Bott index} 
\label{appendixa}

The systems in the DIII symmetry class respect the time-reversal symmetry $T$, particle-hole symmetry $P$, and chiral symmetry $C=-iTP$. Under the eigenbasis of $C$ where $C=\tau_z$, the Bogoliubov–de Gennes Hamiltonian of system can be written as an off-diagonal form
\begin{eqnarray}
H=\begin{pmatrix} 0 & h\\
h^{\dagger}& 0 
\end{pmatrix}.
\label{eq1}
\end{eqnarray}
By the singular value decomposition $h=U_A\Sigma U_B^{\dagger}$, we can define the Bott index \cite{Benalcazar2022,Jiazheng2024}
\begin{eqnarray}
&&N=\frac{1}{2\pi i}\text{Tr}\text{log}(mqm^{\dagger}q^{\dagger}),
\label{seq6}
\end{eqnarray}
where $q=U_AU_B^{\dagger}$ and $m=e^{2\pi i f(\bm r)s_0}$ is a unitary matrix generated by polynomial $f(\bm r)$, with identity matrix $s_0$ acting on the spin space.  

By defining $M=\tau_0\otimes m$ and $Q=\begin{pmatrix}0& q\\
q^{\dagger}&0
\end{pmatrix}$, we have
\begin{eqnarray}
MQM^{\dagger}Q&=&\begin{pmatrix}m& 0\\
0&m\end{pmatrix}\begin{pmatrix}0& q\\
q^{\dagger}&0\end{pmatrix}\begin{pmatrix}m^{\dagger}& 0\\
0&m^{\dagger}\end{pmatrix}\begin{pmatrix}0& q\\
q^{\dagger}&0\end{pmatrix}\nonumber\\
&=&\begin{pmatrix}mqm^{\dagger}q^{\dagger}& 0\\
0&mq^{\dagger}m^{\dagger}q\end{pmatrix}.
\end{eqnarray}
Then we have 
\begin{eqnarray}
N&=&\frac{1}{2\pi i}\text{Tr}\text{log}(mqm^{\dagger}q^{\dagger})\nonumber\\
\label{seq3}
&=&\frac{1}{4\pi i}\text{Tr}\begin{pmatrix}\text{log}(mqm^{\dagger}q^{\dagger})& 0\\
0&-\text{log}(mq^{\dagger}m^{\dagger}q)\end{pmatrix}\\
&=&\frac{1}{4\pi i}\text{Tr}[C\text{log}(MQM^{\dagger}Q)].
\label{seq5}
\end{eqnarray}
In Eq.~\eqref{seq3}, we have used the relation $\text{Trlog}(mqm^{\dagger}q^{\dagger})=-\text{Trlog}(qmq^{\dagger}m^{\dagger})=-\text{Trlog}(mq^{\dagger}m^{\dagger}q)$.
$N$ can be also be defined as
\begin{eqnarray}
N&=&\frac{1}{2\pi i}\text{Tr}\text{log}(mqm^{\dagger}q^{\dagger})\nonumber\\
\label{seqa}
&=&\frac{1}{4\pi i}\text{Tr}\text{log}\begin{pmatrix}mqm^{\dagger}q^{\dagger}& 0\\
0&m^{\dagger}q^{\dagger}mq\end{pmatrix}\\
&=&\frac{1}{4\pi i}\text{Tr}\text{log}(\Xi).
\label{seq4}
\end{eqnarray}
where $\Xi=\frac{\mathbbm{1}+C}{2}{M}Q{M}^{\dagger}Q+\frac{\mathbbm{1}-C}{2}{M}^{\dagger}Q{M}Q$ and we have used the relation $\text{Trlog}(mqm^{\dagger}q^{\dagger})=\text{Trlog}(m^{\dagger}q^{\dagger}mq)$ in Eq.~\eqref{seqa}. Thus, Eqs.~\eqref{seq6}, \eqref{seq5},  and \eqref{seq4} provide three equivalent definitions of the Bott index $N$.

\section{equivalent expressions of the spin Bott indices} 
\label{appendixb}
When a system respects $s_z$ symmetry, namely $[s_z,H]=0$, the system has a new chiral symmetry $\mathcal{C}=s_zC$, which satisfies $\{\mathcal{C},H\}=0$. We define the spin Bott index by the $\mathcal{C}$ operator, 
\begin{eqnarray}
\mathcal{N}&=&\frac{1}{4\pi i}\text{Tr}[\mathcal{C}\text{log}(MQM^{\dagger}Q)]\nonumber\\
&=&\frac{1}{4\pi i}\text{Tr}[s_zC\text{log}(MQM^{\dagger}Q)].
\label{szn}
\end{eqnarray}
 As $[s_z,H]=0$, matrix  $Q$ is block-diagonal in the spin space and can be written as 
\begin{eqnarray}
Q=\begin{pmatrix} Q_{+} & 0\\
0& Q_{-}
\end{pmatrix}.
\label{eq11}
\end{eqnarray}
Thus, $\mathcal{N}$ can be further written as 
\begin{eqnarray}
\mathcal{N}&=&\frac{1}{4\pi i}\text{Tr}[s_zC\text{log}(MQM^{\dagger}Q)\nonumber\\
&=&\frac{1}{4\pi i}\text{Tr}[C_{+}\text{log}(M_{+}Q_{+}M_{+}^{\dagger}Q_{+})-C_{-}\text{log}(M_{-}Q_{-}M_{-}^{\dagger}Q_{-})]\nonumber\\
&=&N_{+}-N_{-},
\label{sp1}
\end{eqnarray}
where matrices $C_{\pm}$ and $M_{\pm}$ are spin resolved and have a dimension that is half of $C$ and $M$. In Eq.~\eqref{sp1}, we have defined
\begin{eqnarray}
\label{eqs:polarized_Bott}
N_{\pm}=\frac{1}{4\pi i}\text{Tr}[C_{\pm}\text{log}(M_{\pm}Q_{\pm}M_{\pm}^{\dagger}Q_{\pm})].
\label{spinbt}
\end{eqnarray}

By utilizing the operator $M_z=e^{2\pi if(\bm r)s_z\tau_0}$, we define another spin Bott index 
\begin{eqnarray}
\bar{\mathcal{N}}=\frac{1}{4\pi i}\text{Tr}[C\text{log}(M_zQM_z^{\dagger}Q)],
\label{sp3}
\end{eqnarray}
With $s_z$ conservation,  matrices $C,Q,$ and $M_z$ are block-diagonal in spin space and we have 
\begin{widetext}
\begin{eqnarray}
\bar{\mathcal{N}}&=&\frac{1}{4\pi i}\text{Tr}[C\text{log}(M_zQM_z^{\dagger}Q)\nonumber\\
&=&\frac{1}{4\pi i}\begin{pmatrix} \text{Tr}[C_{+}\text{log}(M_{+}Q_{+}M_{+}^{\dagger}Q_{+})]&0\\
0&\text{Tr}[C_{-}\text{log}(M_{-}^{\dagger}Q_{-}M_{-}Q_{-})]\end{pmatrix}\nonumber\\
&=&{N}_{+}-{N}_{-}\nonumber\\
&=&\mathcal{N},
\end{eqnarray}
\end{widetext}
where $M_z=\begin{pmatrix} M_{+} & 0\\
0& M_{-}
\end{pmatrix}
\label{eq}
$. In the third equality,  we have used the relation 
\begin{eqnarray}
N_{-}&=&\frac{1}{4\pi i}\text{Tr}[C_{-}\text{log}(M_{-}Q_{-}M_{-}^{\dagger}Q_{-})]\nonumber\\
&=&\frac{1}{2\pi i}\text{Tr}\text{log}(m_{-}q_{-}m_{-}^{\dagger}q_{-}^{\dagger})\nonumber\\
&=&-\frac{1}{2\pi i}\text{Tr}\text{log}(m_{-}^{\dagger}q_{-}m_{-}q_{-}^{\dagger})\nonumber\\
&=&-\frac{1}{4\pi i}\text{Tr}[C_{-}\text{log}(M_{-}^{\dagger}Q_{-}M_{-}Q_{-})].
\end{eqnarray}
Here, $m_{-}$ and $q_{-}$ are spin resolved and have a dimension that is half of $m$ and $q$.

Besides Eq.~\eqref{szn} and Eq.~\eqref{sp3},
there is a third definition of spin Bott index by using the spin projection. 
We first construct the spin projection operator \cite{Prodan2009}
\begin{eqnarray}
P_z=Ps_zP, \quad P=(\mathbbm{1}-Q)/2,
\end{eqnarray}
 Then we can further define the new projection operators and flatten Hamiltonians
\begin{eqnarray}
P_{\pm}=\sum_{i}|\phi_{i}^{\pm}\rangle\langle \phi_{i}^{\pm}|, \quad  \tilde{Q}_{\pm}=-P_{\pm}+CP_{\pm}C,
\label{eqe2}
\end{eqnarray}
where $|\phi_{i}^{+}\rangle$ ($|\phi_{i}^{-}\rangle$) is the $i$th eigenstate of $P_z$ with positive (negative) eigenvalue. We define the spin Bott index 
\begin{eqnarray}
\label{eqs:s_spinbott}
&\tilde{\mathcal{N}}_{\pm}&=\frac{1}{4\pi i}\operatorname{Tr}[C\operatorname{log}(M\tilde{Q}_{\pm}M^{\dagger}\tilde{Q}_{\pm})]=\operatorname{Bott}(m,\tilde{q}_{\pm}),\nonumber\\
&\tilde{\mathcal{N}}&=\tilde{\mathcal{N}}_{+}-\tilde{\mathcal{N}}_{-},
\label{pn}
\end{eqnarray}
 where $\tilde{q}_{\pm}$ is obtained by expressing $\tilde{Q}_{\pm}$ with $$\tilde{Q}_{\pm}=\begin{pmatrix}
     0&\tilde{q}_{\pm} \\
     \tilde{q}_{\pm}^{\dagger} &0
 \end{pmatrix}.$$
 When $[s_z,H]=0$, the eigenvalues of $P_z$ consist of just two types of
nonzero values, $\pm 1$. When projecting onto the spin up and down spaces, we have $\tilde{Q}_{\pm}=Q_{\pm}$, giving rise to $\tilde{\mathcal{N}}_{\pm}={\mathcal{N}}_{\pm}$ and $\tilde{\mathcal{N}}={\mathcal{N}}=\bar{\mathcal{N}}$.  Thus, Eqs.~\eqref{szn}, \eqref{sp3}, and \eqref{pn} provide three equivalent expressions of the spin Bott index when  $[s_z,H]=0$.

\section{Alternative choice of polynomials $f_6^{(i)}$}
\label{appendixp}
The choice of polynomials $f_p^{(i)}$ is not unique. For example, when $p=6$, besides the given expression in Eq.~\eqref{fe}, $f_6^{(i)}$ can be alternatively chosen as
\begin{eqnarray}
&&f_6^{(1)}(r)=\frac{\left(x^3-\frac{8 x^2 y}{\sqrt{3}}-3 x y^2+\frac{8 y^3}{3 \sqrt{3}}\right)}{2 L^3}, \nonumber\\
&&f_6^{(2)}(r)=\frac{\left(x^3-\frac{8 x^2 y}{\sqrt{3}}-\frac{x y^2}{3}\right)}{2 L^3}, \nonumber\\
&&f_6^{(3)}(r)=\frac{\left(x^3-\frac{8 x^2 y}{\sqrt{3}}-\frac{x y^2}{3}+\frac{16 y^3}{3 \sqrt{3}}\right)}{2 L^3}, \nonumber\\
&&f_6^{(4)}(r)=\frac{\left(x^2-\frac{5 y^2}{3}\right)}{2 L^2},\nonumber\\
&&f_6^{(5)}(r)=\frac{\left(x^2-\frac{4 x y}{\sqrt{3}}-\frac{y^2}{3}\right)}{2 L^2},
\end{eqnarray}
The above polynomials give rise to $\text{det}(M)\neq 0$.

\section{Robustness of the spin Bott index}
\label{appendixc}
\label{supIII}
In this section, we show that the Bott indices $\bar{\mathcal{N}}$ and $\tilde{\mathcal{N}}$, obtained through Eq.~\eqref{sp3} (Eq.~(6) in the main text) and Eq.~\eqref{pn}, respectively, are robust, and $\bar{\mathcal{N}}=\tilde{\mathcal{N}}$,  unless the energy gap of the system closes or the gap of $P_z$ closes.

\subsection{Notations}
$\sigma(\cdot)$ denotes the set of eigenvalues of a matrix.

$\sigma_{max}(\cdot)$ denotes the largest eigenvalue of a matrix.

$\underset{S}{\text{sup}} P$ represents the supremum of the values taken by $P$ over a set $S$. 

$\parallel\cdot\parallel$ denotes the spectral norm of a matrix (the largest singular value of a matrix). This norm is induced by the Euclidean norm, $|\cdot|$, for vectors and is given by $\parallel A\parallel=\underset{x\neq0}{\text{sup}}\frac{|Ax|}{|x|}$, where $x$ is a vector.

For the spectral norm, we have the following two inequalities for two square matrices $A$ and $B$,
\begin{eqnarray}
    \parallel A + B \parallel &\le \parallel A \parallel + \parallel B \parallel \label{eqs:norminequal_1}\\
    \parallel A  B \parallel &\le \parallel A \parallel  \parallel B \parallel \label{eqs:norminequal_2}.
\end{eqnarray}

$\operatorname{dist}(n,m)$ represents the Euclidean distance function between two vectors denoted by $n$ and $m$ in the position space.

$\mathcal{O}$ denotes the order of approximation.

\subsection{Mathematical derivation}

First, we consider the spin Bott index in Eq.~\eqref{sp3}. We have
\begin{equation}
    \bar{\mathcal{N}}=\frac{1}{4\pi i}\text{Tr}[C\text{log}(M_zQM_z^{\dagger}Q)]=\operatorname{Bott}(m_z,q),
\end{equation}
where $m_z=e^{2\pi i f(\bm r) s_z}$. It follows that
\begin{widetext}
\begin{equation}
\begin{aligned}
   m_z &= e^{2\pi i f(\bm r) s_z }\\
    &=e^{2\pi i f(\bm r) (1-2c_z) } \\
    &=e^{2\pi i f ( \bm r)} \left(1+\left(-4\pi i f(\bm r)c_z\right)+\left(-4\pi i f(\bm r)c_z\right)^2/2!+\cdots\right) \\
    &=e^{2\pi i f ( \bm r)}\left(1-c_z+c_z\left(1+\left(-4\pi i f(\bm r)\right)+\left(-4\pi i f(\bm r)\right)^2/2!+\cdots\right)\right) \\
    &=e^{2\pi i f(\bm r) }(c_ze^{-4\pi i f(\bm r)  }+1-c_z)\\
    &=c_ze^{-2\pi i f(\bm r) }+(1-c_z)e^{2\pi i f(\bm r) },
\end{aligned}
\end{equation}
\end{widetext}
where $c_z$ is the Fermi projector of $s_z$, which implies that $s_z=1-2c_z$ and $c_z^2=c_z$.  In the third step of the above derivation, we utilize the equality $c_z^n=c_z$ for $n \ge 1$.

We introduce the following theorem.
\begin{theorem}
    \label{thm: homoinvar}
    Given two continuous maps $V(s)$ : $ [0,1]\to \mathcal{U}(N)$ and $W(s)$ : $ [0,1]\to \mathcal{U}(N)$, where $\mathcal{U}(N)$ represents the unitary group, with $V(0)=V$, $W(0)=W$, such that $\parallel [V(s),W(s)]\parallel <2,\forall s\in [0,1]$, then
    \begin{equation}
        \label{eqs: homoinvar}
        \operatorname{Bott}\left(V(s),W(s)\right)=\operatorname{Bott}(V,W).
    \end{equation}
\end{theorem}
\begin{proof}
    The proof of this theorem can be found in Refs.~\cite{BottIndexTwo2021toniolo,BottIndexUnitary2022toniolo}.
\end{proof}

This theorem provides the requirement for the Bott index to remain the same, which is key for us to demonstrate the robustness of $\bar{\mathcal{N}}$ and $\tilde{\mathcal{N}}$, as well as $\bar{\mathcal{N}}=\tilde{\mathcal{N}}$. Thus, let us consider the norm of $[m_z,q]$. We have the following inequalities
\begin{widetext}
\begin{equation}
\begin{aligned}
    \parallel [m_z,q] \parallel &\le \parallel [c_ze^{-2\pi i f(\bm r) },q] \parallel+\parallel [(1-c_z)e^{2\pi i f(\bm r) },q] \parallel\\
    &\le \parallel [c_z,q]e^{-2\pi i f(\bm r) } \parallel+\parallel c_z[e^{-2\pi i f(\bm r) },q] \parallel+\parallel [1-c_z,q]e^{2\pi i f(\bm r) } \parallel+\parallel (1-c_z)[e^{2\pi i f(\bm r) },q] \parallel\\
    & \le 2\parallel [c_z,q] \parallel + \parallel[e^{2\pi i f(\bm r) },q] \parallel+\parallel[e^{-2\pi i f(\bm r) },q] \parallel.
\end{aligned}
\end{equation}
\end{widetext}
We introduce the following theorem:
\begin{theorem}
    \label{thm: projector_norm}
    Given a chiral-symmetric Hamiltonian 
    \begin{equation}
      H=\begin{pmatrix} 0 & h\\
h^{\dagger}& 0 
\end{pmatrix},
    \end{equation} 
    with finite coupling $R$ and a spectral gap $\Delta E$ for all states except those that are localized at corners, $q$ is defined as follows:
    \begin{equation}
        q=U_A U_B^{\dagger},
    \end{equation}
    where we use the singular value decomposition   $h=U_A \Sigma U_B^{\dagger}$. The following relationship holds true,
    \begin{equation}
    \begin{aligned}
        &m=e^{2\pi i \frac{f(X,Y,Z,\dots)}{g(L)}},\\
&\parallel[m,q]\parallel \le \mathcal{O}\left(\frac{R}{L} \frac{\parallel H\parallel}{\Delta E} \right),
    \end{aligned}
    \end{equation}
    for polynomials $f$ and $g$ with $\operatorname{deg}(f)=\operatorname{deg}(g)$ and $f(\boldsymbol{x}_j)=\pm 1/2$ where $\boldsymbol{x}_j$ is the position of the $j$th corner.
\end{theorem}

\begin{proof}
   A weak version of this theorem with the same conclusion has been proven in our joint work~\cite{Jiazheng2024} for Hamiltonians with an energy gap for all states. Therefore, we only need to address situations where gapless corner states appear. We have
   \begin{equation}
       \begin{aligned}
           \parallel[m,q]\parallel &=\parallel m q m^{\dagger} q^{\dagger}-1\parallel \\
           &= \parallel U_A^{\dagger}m U_A U_B^{\dagger} m^{\dagger} U_B-1\parallel.
       \end{aligned}
   \end{equation}
   It has been proven that $U_A^{\dagger}m U_A U_B^{\dagger} m^{\dagger} U_B$ is a block-diagonal matrix in Ref.~\cite{Jiazheng2024}, where the equality 
   \begin{equation}
U_{A,\text{corner}}^{\dagger}m U_{A,\text{corner}}U_{B,\text{corner}}^{\dagger}m^{\dagger}U_{B,\text{corner}}=\mathbbm{1}
   \end{equation}
   has also been proven. $U_{A,\text{corner}}$ denotes the matrix composed of the $A$-subspace component of eigenstates residing in corners. Utilizing this result, we have
\begin{widetext}
   \begin{equation}
   \begin{aligned}
       \parallel U_A^{\dagger}m U_A U_B^{\dagger} m^{\dagger} U_B-1\parallel &= \parallel \bigoplus_{\beta\in \{\text{bulk},\text{edge},\text{corner}\}}U_{A,\beta}^{\dagger}m U_{A,\beta}U_{B,\beta}^{\dagger}m^{\dagger} U_{B,\beta} -1 \parallel\\ &= \operatorname{sup}_{\beta\in \{\text{bulk},\text{edge},\text{corner}\}}(\parallel U_{A,\beta}^{\dagger}m U_{A,\beta}U_{B,\beta}^{\dagger}m^{\dagger} U_{B,\beta} -1\parallel)\\&=\operatorname{sup}_{\beta\in \{\text{bulk},\text{edge}\}}(\parallel U_{A,\beta}^{\dagger}m U_{A,\beta}U_{B,\beta}^{\dagger}m^{\dagger} U_{B,\beta} -1\parallel).
   \end{aligned}
   \end{equation}
\end{widetext}
   We introduce an effective Hamiltonian $\tilde{H}$, which is composed of edge and bulk states and features both a spectral gap and a finite coupling range. The finite coupling range of $\tilde{H}$ is inherited from the finite coupling of $H$, as evidenced by the representation $H=\sum_{j\notin \text{corner}}E_j |\Psi_j\rangle \langle \Psi_j |$ and $\tilde{H}=R\cdot H \cdot R^{T}$, where $R$ is a rectangular matrix that projects onto the Hilbert space composed by edge and bulk states. Since this effective Hamiltonian features a gap for all states and the finite coupling range, we apply the weak version of the theorem to this Hamiltonian $\tilde{H}$. It follows that 
\begin{widetext}
   \begin{equation}
       \operatorname{sup}_{\beta\in \{\text{bulk},\text{edge}\}}(\parallel U_{A,\beta}^{\dagger}m U_{A,\beta}U_{B,\beta}^{\dagger}m^{\dagger} U_{B,\beta} -1\parallel) \le \mathcal{O}\left(\frac{R}{L} \frac{\parallel H\parallel}{\Delta E} \right).
   \end{equation}
\end{widetext}
   Thus, we have proven that
   \begin{equation}
       \parallel[m,q]\parallel \le \mathcal{O}\left(\frac{R}{L} \frac{\parallel H\parallel}{\Delta E} \right).
   \end{equation}
\end{proof}

According to this theorem, when a spectral gap $\Delta E$ exists for all states except those that are localized at corners, we have
\begin{equation}
\begin{aligned}
     \parallel [m_z,q] \parallel & \le 2\parallel [c_z,q] \parallel + \mathcal{O}\left(\frac{R}{L} \frac{\parallel H\parallel}{\Delta E} \right) \\
     & \le 2 \parallel [(1-s_z)/2,2P] \parallel + \mathcal{O}\left(\frac{R}{L} \frac{\parallel H\parallel}{\Delta E} \right) \\
    & \le 2 \parallel [s_z,P] \parallel + \mathcal{O}\left(\frac{R}{L} \frac{\parallel H\parallel}{\Delta E} \right),
\end{aligned}
\end{equation}
where 
\begin{equation}
     P=\left( \begin{matrix}
	\frac{1}{2}&		-\frac{q}{2}\\
	-\frac{q^{\dagger}}{2}&		\frac{1}{2}\\
\end{matrix} \right).
\end{equation}

Next, we prove the equivalence between $\parallel [s_z,P] \parallel<1$ and the existence of a gap of $P_z=Ps_zP$. 

Rewriting all operators in the eigenbasis of the Hamiltonian, we have
\begin{equation}
    P=\begin{pmatrix}
        1 & 0 \\
        0& 0
    \end{pmatrix},
\end{equation}
and
\begin{equation}
    s_z=\begin{pmatrix}
        A & D \\
        D^{\dagger}& B
    \end{pmatrix},
\end{equation}
where $A=A^{\dagger}$ and $B=B^{\dagger}$. Since $s_z^2=1$, it follows that
\begin{equation}
\begin{aligned}
      A^2+DD^{\dagger}&=1,\\
      B^2+D^{\dagger}D&=1, \\
      AD + DB &=0.
\end{aligned}
\end{equation}
We have 
\begin{equation}
    [s_z,P]=\begin{pmatrix}
        0 & -D \\
        D^{\dagger} & 0
    \end{pmatrix},
\end{equation}
and 
\begin{equation}
    P_z=Ps_zP=\begin{pmatrix}
        A & 0 \\
        0 & 0
    \end{pmatrix}.
\end{equation}
The norm of $[s_z,P]$ is equal to the largest eigenvalue of $[s_z,P]^{\dagger}[s_z,P]$.
\begin{equation}
    [s_z,P]^{\dagger}[s_z,P]=\begin{pmatrix}
        DD^{\dagger} & 0\\
        0 & DD^{\dagger}
    \end{pmatrix}.
\end{equation}
It follows that
\begin{equation}
    \sigma_{max}( [s_z,P]^{\dagger}[s_z,P] ) = \sigma_{max}( DD^{\dagger} ).
\end{equation}
Denoting the eigenvector with the largest eigenvalue of $DD^{\dagger}$ as $|\psi_{max}\rangle$, we have
\begin{equation}
\langle\psi_{max}|A^2+DD^{\dagger}|\psi_{max}\rangle=1=\langle\psi_{max}|A^2|\psi_{max}\rangle+\sigma_{max}.
\end{equation}
Noting that $A^2$ is positive semi-definite, we have 
\begin{equation}
    \begin{aligned}
        &\sigma_{max}=1-\langle \psi_{max}|A^2|\psi_{max}\rangle \le 1, \\
        &\parallel [s_z,P] \parallel \le 1.
    \end{aligned}
\end{equation}

Equality holds if and only if $\det(A)=0$. Thus, we show the equivalence between the existence of a gap of $P_z$ ($\det(A)\neq0$) and $\parallel [s_z,P] \parallel<1$, implying that the presence of both a gap of $P_z$ at zero and $\Delta E$ ensures that
\begin{equation}
    \parallel [m_z,q] \parallel<2.
\end{equation}
According to Theorem~\ref{thm: homoinvar}, we have demonstrated that the Bott index $\bar{\mathcal{N}}$ is robust with the existence of both a gap of $P_z$ at zero and $\Delta E$.

\begin{figure*}
\centering
\includegraphics[width=0.8\textwidth]{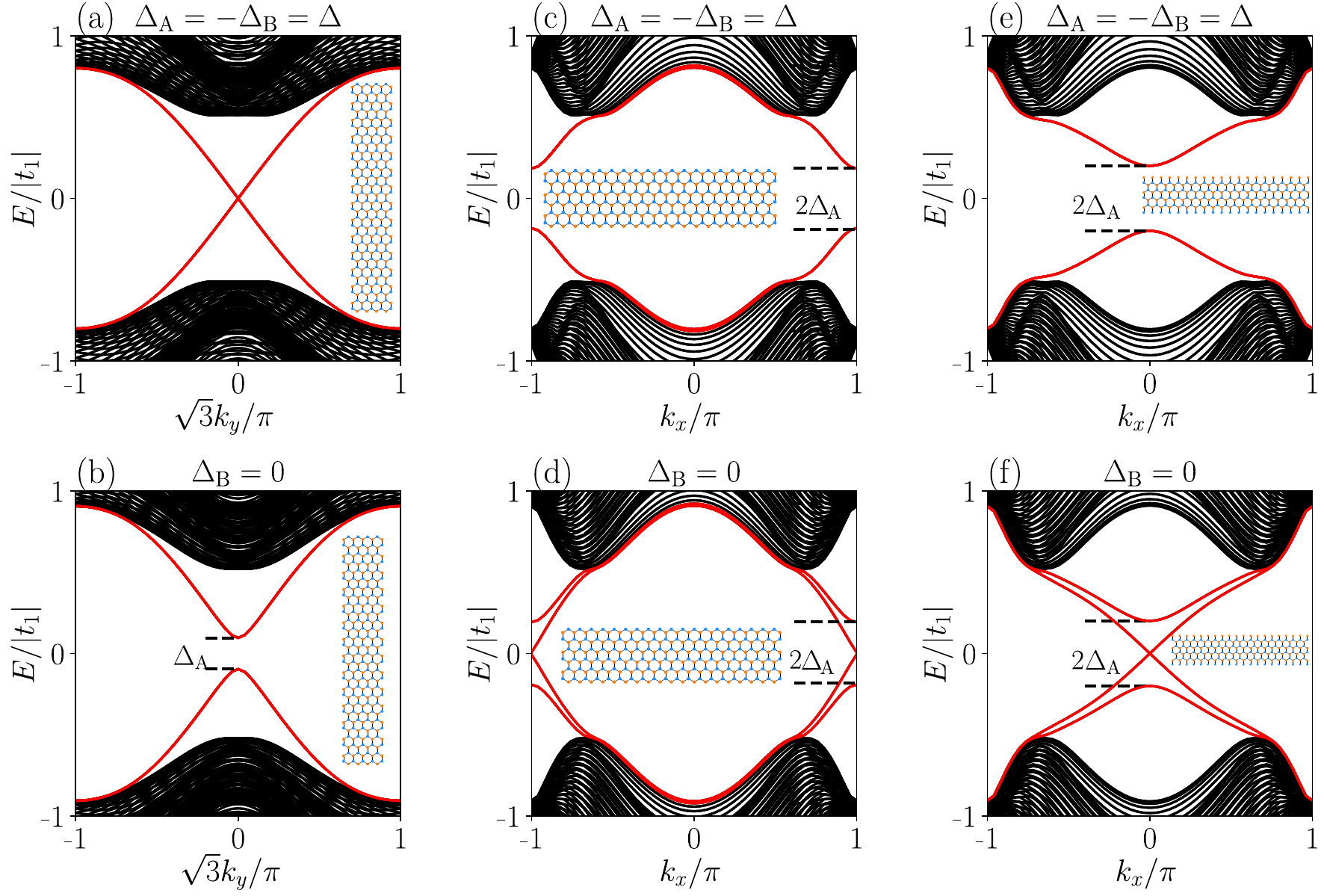}
\caption{Energy spectra for a ribbon with armchair edge in (a) and (b), zigzag edge in (c) and (d), and bearded edge in (e) and (f). The common parameters are taken as $t=1,\lambda_{\text{so}}=0.1,\Delta_{\text{A}}=0.2$, and $\mu=\lambda_v=\lambda_{\text{R}}=0$. }
\label{edge spectra}
\end{figure*}

Next, we consider the spin Bott index defined by Eq.~\eqref{pn}.

Utilizing Theorem~\ref{thm: projector_norm}, we can prove the following theorem:
\begin{theorem}
    \label{thm: spin_projector_norm}
    Given a Hamiltonian $H$ with finite coupling $R$ and a spectral gap $\Delta E$ for all states except states that are localized at corners, the spin-polarized projector $P_{\pm}$ is defined as follows:
    \begin{equation}
        P_{\pm}=\sum_{i}|\phi_{i}^{\pm}\rangle\langle \phi_{i}^{\pm}|,
    \end{equation}
    where $|\phi_{i}^{\pm}\rangle$ is the $i$th eigenstate of $P_z=Ps_zP$ with positive (negative) eigenvalues. If $P_z$ possesses a spectral gap $\Delta E_z$, the following relationship holds true,
    \begin{equation}
    \begin{aligned}
    &m=e^{ 2 \pi i \frac{f(X,Y,Z,\dots)}{g(L)}}, \\
        &\parallel[m,P_{\pm}]\parallel  \le \mathcal{O}\left(\frac{R}{L} \frac{\parallel H\parallel}{\Delta E \Delta E_z} \right),
    \end{aligned}
    \end{equation}
    for polynomials $f$ and $g$ with $\operatorname{deg}(f)=\operatorname{deg}(g)$ and $f(\boldsymbol{x_j})=\pm1/2$ where $\boldsymbol{x_j}$ denote the position of all corners.
\end{theorem}
For convenience, we illustrate the proof for $P_{-}$, noting that $P_{+}$ can be handled analogously.
\begin{proof}
    The spin-polarized projector $P_{-}$ can be expressed as:
    \begin{equation}
    P_{-}=\frac{1}{2 \pi i} \oint_{\Gamma} dz \left(z-P_z\right)^{-1},
    \end{equation}
    where $\Gamma$ encloses the negative eigenvalues of $P_z$ in the complex plane.
    For $z \notin \sigma(P_z)$ and any matrix $A$ with the same size as $P_z$, we have the equality
    \begin{equation}
    \label{eqs:commutor_equa}
        \begin{aligned}
\left[A,(P_z-z)^{-1}\right] &=(P_z-z)^{-1}\left[(P_z-z),A\right](P_z-z)^{-1}  \\ &=(P_z-z)^{-1}\left[P_z,A\right](P_z-z)^{-1}.
\end{aligned}
    \end{equation}
    Using the above results, we have 
    \begin{equation}
        \parallel [m,P_{-}] \parallel \le \frac{1}{2 \pi}\parallel[m,P_z]\parallel \oint_{\Gamma} \parallel\left(P_z-z\right)^{-1}\parallel^2 |dz|,
    \end{equation}
with $\parallel\left(P_z-z\right)^{-1}\parallel^{2}= [\operatorname{dist}\left(z, \sigma(P_z)\right)]^{-2}$. $\operatorname{dist}\left(z, \sigma(P_z)\right)$ denotes the distance from a point $z$ to the region $\sigma(P_z)$. Taking the radius of $\Gamma$ to $\infty$, the loop-integral becomes
\begin{widetext}
    \begin{equation}
    \label{eqs:loop-int}
    \oint_{\Gamma} \parallel\left(P_z-z\right)^{-1}\parallel^2 |dz|=\oint_{\Gamma} [\operatorname{dist}\left(z, \sigma(P_z)\right)]^{-2} |dz|=\int_{-\infty}^{\infty} \frac{1}{(\frac{\Delta E_z}{2})^2+(y)^2}dy=\frac{2 \pi}{\Delta E_z},
\end{equation}
\end{widetext}
where $y$ denotes $\operatorname{Im}z$.
Also, we have
\begin{equation}
\begin{aligned}
    \parallel[m,P_z]\parallel &=\parallel [m,P]s_zP+Ps_z[m,P]\parallel \\ &\le 2\parallel[m,P] \parallel\le \mathcal{O}\left(\frac{R}{L} \frac{\parallel H\parallel}{\Delta E} \right),
    \end{aligned}
\end{equation}
where we use two inequalities of matrix norm Eqs.~(\ref{eqs:norminequal_1}, \ref{eqs:norminequal_2}), and in the last step we use Theorem~\ref{thm: projector_norm}.

Using the above findings, we have
\begin{equation}
        \parallel[m,P_{-}]\parallel\le \mathcal{O}\left(\frac{R}{L} \frac{\parallel H\parallel}{\Delta E \Delta E_z} \right),
\end{equation}
where $\mathcal{O}$ denotes the order of approximation.
\end{proof}

Applying this theorem to Eq.~\eqref{eqe2}, we have
\begin{equation}
\begin{aligned}
\label{eqs:s_minequa}
    &\parallel [M,\tilde{Q}_{\pm}] \parallel \le \mathcal{O}\left(\frac{R}{L} \frac{\parallel H\parallel}{\Delta E \Delta E_z} \right) ,\\
    &\parallel [m,\tilde{q}_{\pm}] \parallel \le \mathcal{O}\left(\frac{R}{L} \frac{\parallel H\parallel}{\Delta E \Delta E_z} \right)  .      
\end{aligned}
\end{equation}

In the following, we show that $\bar{\mathcal{N}}=\tilde{\mathcal{N}}$ unless the energy gap closes or the gap of $P_z$ closes by utilizing the above results.

Consider a Hamiltonian 
\begin{equation}
    H=H_{\text{cons}}+H_{\text{brok}},
\end{equation}
where $H_{\text{cons}}$ is the Hamiltonian with $s_z$ conserved and $H_{\text{brok}}$ denotes the term that breaks $s_z$ symmetry, namely $[s_z,H_{\text{brok}}]\neq 0$. 
We introduce a interpolating Hamiltonian 
\begin{equation}
    H(s)=H_{\text{cons}}+s H_{\text{brok}},\quad s\in [0,1].
\end{equation}
Since for $s\in [0,1]$, the energy gap $\Delta E$ and $\Delta E_z$ remain finite, we have
\begin{equation}
\begin{aligned}
        &\parallel [m_z,q(s)] \parallel&<2,\\
        &\parallel [m,\tilde{q}_{\pm}(s)] \parallel&<2,
\end{aligned}
\end{equation}
when $L\to \infty$. $q(s)$ and $\tilde{q}_{\pm}(s)$ are obtained by replacing $H$ with $H(s)$. 
Applying Theorem~\ref{thm: homoinvar} to $\bar{N}(s)=\operatorname{Bott}(m_z,q(s))$ and $\tilde{N}_{\pm}(s)=\operatorname{Bott}(m,\tilde{q}_{\pm}(s))$, we have
\begin{equation}
    \begin{aligned}
        &\bar{N}(1)=\bar{N}(0),\\
        &\tilde{N}_{\pm}(1)=\tilde{N}_{\pm}(0).
    \end{aligned}
\end{equation}
Since when $s=0$ ($[m,H(0)]=0$) we have $\bar{N}(0)=\tilde{N}(0)=\tilde{N}_{+}(0)-\tilde{N}_{-}(0)$, it follows that
\begin{equation}
\begin{aligned}
    \bar{N}(1)=\bar{N}(0)&=\tilde{N}(0)\nonumber\\
    &=\tilde{N}_{+}(0)-\tilde{N}_{-}(0)\nonumber\\
    &=\tilde{N}_{+}(1)-\tilde{N}_{-}(1)\nonumber\\
    &=\tilde{N}(1).
\end{aligned}
\end{equation}
Hence, we have shown that the Bott indices $\bar{\mathcal{N}}$ and $\tilde{\mathcal{N}}$ are robust and $\bar{\mathcal{N}}=\tilde{\mathcal{N}}$ unless the energy gap of the system closes or the gap of $P_z$ closes.

\begin{figure}
\centering
\includegraphics[width=0.5\textwidth]{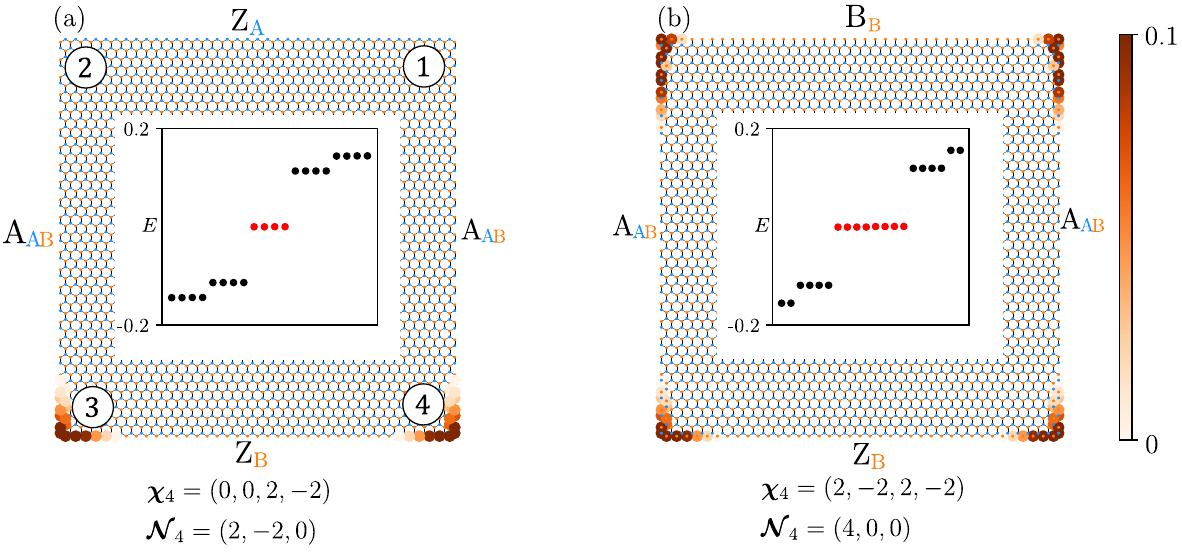}
\caption{Different patterns of MKP on the square-shaped systems. The circled number in (a) at corners labels the order of corners. The letters $\text{A}_{\text{AB}}$,  $\text{Z}_{\alpha}$,  $\text{B}_{\alpha}$ label the armchair, zigzag, and bearded edges, respectively, with $\alpha\in$ \{A, B\}.  The insets plot the energies close to zero. The vectors $\bm{\chi}_4$ and $\bm{\mathcal{N}}_4$ are used to describe and characterize the different patterns of MKP. The model parameters are taken as $t=1,\lambda_{so}=0.2,\Delta_{\text{A}}=0.5,\Delta_{\text{B}}=-0.25$.  }
\label{sq}
\end{figure}

\section{Edge theory }
\label{appendixd}
For simplicity in our edge theory analysis, we take $\lambda_\text{R}=\lambda_v=0$ in the Hamiltonian of Kane-Mele model with sublattice-dependent pairing potential. In the Nambu basis $\Psi_{\bm k}=(c_{\text{A},\bm k,\uparrow},c_{\text{A},\bm k,\downarrow},c_{\text{B},\bm k,\uparrow},c_{\text{B},\bm k,\downarrow}, c_{\text{A},-\bm k,\downarrow}^{\dagger},-c_{\text{A},-\bm k,\uparrow}^{\dagger},c_{\text{B},-\bm k,\downarrow}^{\dagger},-c_{\text{B},-\bm k,\uparrow}^{\dagger})$,
the Bloch Hamiltonian of the  Kane-Mele model with sublattice-dependent superconducting pairings can be  written as 
\begin{eqnarray}
&&{H}(\bm k)=f_x(\bm k)\tau_z\sigma_xs_0+f_y(\bm{k})\tau_z\sigma_ys_0+f_z(\bm{k})\tau_z\sigma_zs_z+\nonumber\\
&&\quad\quad\quad\quad\Delta_{\text{A}}\tau_x(\sigma_0+\sigma_z)/2s_0+\Delta_{\text{B}}\tau_x(\sigma_0-\sigma_z)/2s_0,\nonumber\\
&& f_x(\bm k)=t(1+\cos (k_x/2)\cos(\sqrt{3}k_y/2)),\nonumber\\
&&f_y(\bm k)=t\cos (k_x/2)\sin(\sqrt{3}k_y/2), \nonumber\\
&&f_z(\bm k)=\lambda_{\text{so}}(2\sin (k_x)-4\sin (k_x/2)\cos(\sqrt{3}k_y/2)).
\label{hab}
\end{eqnarray}
where Pauli matrices $\tau$, $\sigma$, and $s$ act on the particle-hole, orbital, and spin space, respectively. Here, we have chosen the lattice constant $a=1$.

\begin{figure*}
\centering
\includegraphics[width=0.8\textwidth]{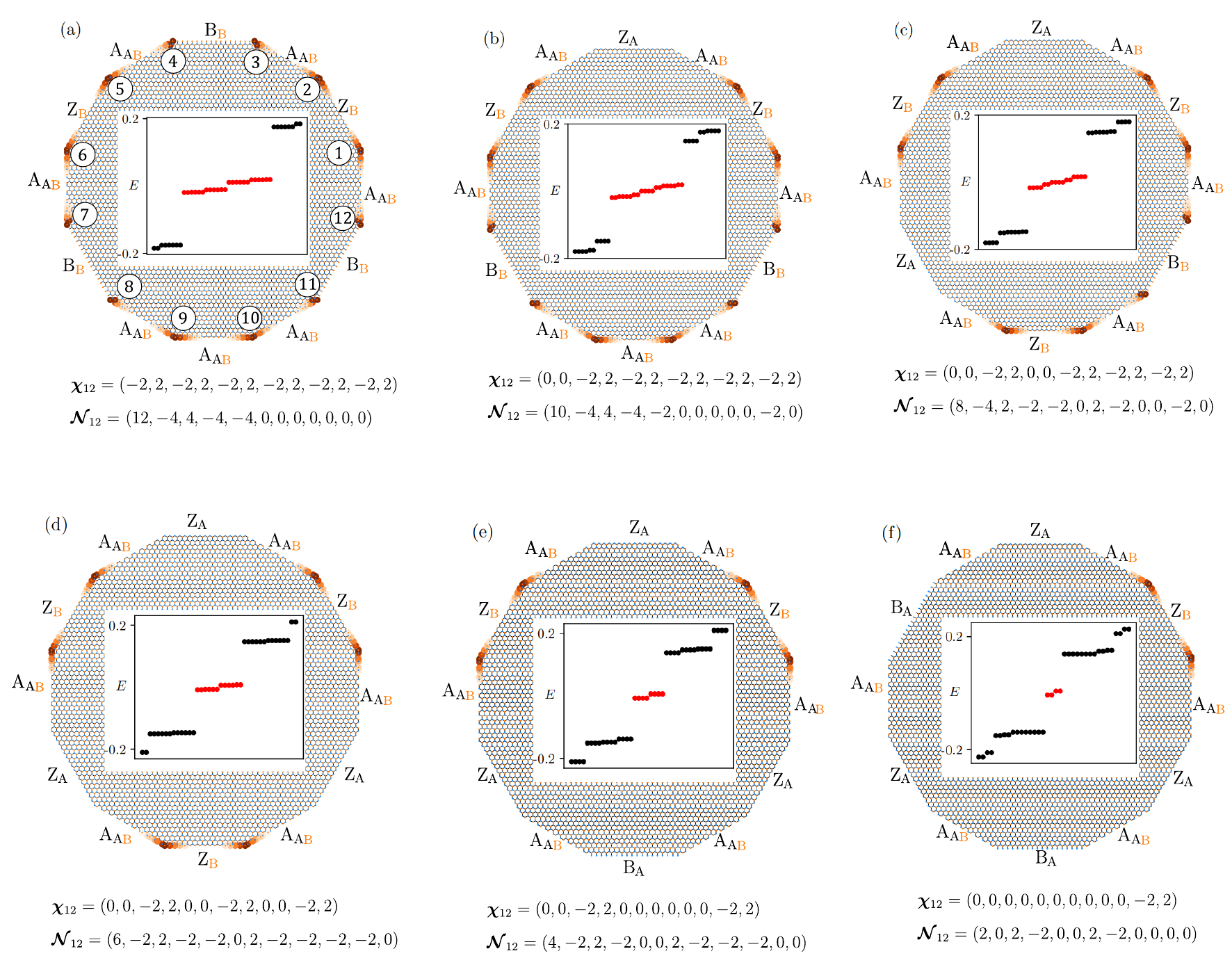}
\caption{(a-f) Different patterns of MKP on dodecagonal-shaped systems. 
The letters $\text{A}_{\text{AB}}$,  $\text{Z}_{\alpha}$,  $\text{B}_{\alpha}$ label the armchair, zigzag, and bearded edges, respectively, with $\alpha\in$ \{A, B\}. The circled number in (a) at corners labels the order of corners. The insets plot the energies close to zero. The vectors $\bm{\chi}_{12}$ and $\bm{\mathcal{N}}_{12}$ are used to describe and characterize the different patterns of MKP.  The same model parameters are used as those in Fig.~\ref{sq}.}
\label{12m}
\end{figure*}

For the armchair edge, the boundary Dirac point is located at $k_y=0$. Therefore, we expand 
${H}$ around $k_y=0$ (the expansion is only performed in the $k_y$ direction), leading to ${H}={H}_0+{H}_1$, where
\begin{widetext}
\begin{eqnarray}
&&{H}_0=t(1+\cos (k_x/2))\tau_z\sigma_xs_0+\lambda_{\text{so}}(2\sin (k_x)-4\sin (k_x/2))\tau_z\sigma_zs_z,\nonumber\\
&&{H}_1= \Delta_{\text{A}}\tau_x(\sigma_0+\sigma_z)/2s_0+\Delta_{\text{B}}\tau_x(\sigma_0-\sigma_z)/2s_0+\sqrt{3}t k_y\cos (k_x/2)\tau_z\sigma_ys_0.
\end{eqnarray}
\end{widetext}
The spinor part of the zero-energy states of ${H}_0$ for the semi-infinite system with $x \in (-\infty, 0)$ 
takes the form \cite{Pan2019}
\begin{eqnarray}
&& \psi_1=|\tau_z=1\rangle\otimes|\sigma_y=1\rangle\otimes|s_z=1\rangle, \nonumber\\
&&  \psi_2=|\tau_z=-1\rangle\otimes|\sigma_y=1\rangle\otimes|s_z=-1\rangle,\nonumber\\
&& \psi_3=|\tau_z=-1\rangle\otimes|\sigma_y=-1\rangle\otimes|s_z=1\rangle,\nonumber\\
&& \psi_4=|\tau_z=1\rangle\otimes|\sigma_y=-1\rangle\otimes|s_z=-1\rangle.
\end{eqnarray}
The $k_y$-dependent term couples
the zero-energy states away from $k_y = 0$, which gives rise to the dispersion of edge state. At exact Dirac point $k_y=0$, the above four zero-energy states are coupled by the term $\Delta_{\text{A}}/2\tau_x\sigma_0s_0+\Delta_{\text{B}}/2\tau_x\sigma_0s_0$, which leads to the superconducting gap with magnitude $(\Delta_{\text{A}}+\Delta_{\text{B}})/2$. This is consistent with the numerical results shown in  Fig.~\ref{edge spectra} (a) and \ref{edge spectra}(b).

For the zigzag edge, the boundary Dirac point is located at $k_x=\pi$. Therefore, we expand ${H}$ around at $k_x=\pi$ leading to ${H}={H}_0+{H}_1$, where
\begin{eqnarray}
&&{H}_0=t(1-\cos (\sqrt{3}k_y/2))\tau_z\sigma_xs_0-t\sin (\sqrt{3}k_y/2)\tau_z\sigma_ys_0,\nonumber\\
&&{H}_1= -2\lambda_{\text{so}}(\delta k_x+2\cos(\delta_{k_x})\cos(\sqrt{3}k_y/2))\tau_z\sigma_zs_z\nonumber\\
&&\quad\quad\quad+\Delta_{\text{A}}\tau_x(\sigma_0+\sigma_z)/2s_0+\Delta_{\text{B}}\tau_x(\sigma_0-\sigma_z)/2s_0.
\end{eqnarray}
Here $\delta k_x =k_x -\pi$. As ${H}_0$ only contains $\sigma_x$ and $\sigma_y$ in the sublattice space, the zero-energy states of ${H}_0$ 
are the eigenstates of $|\sigma_z= 1\rangle$ ($|\sigma_z= -1\rangle$)
when the zigzag edge is formed by the A (B) atom. The 
 term $-4\lambda_{\text{so}}\cos(\delta_{k_x})\cos(\sqrt{3}k_y/2))\tau_z\sigma_zs_z$ can not remove the zero-energy states but modify the spinor part of the zero-energy states. When $\lambda_{\text{so}}/t$ is small, the zero-energy states approximately take the form $|\sigma_z= 1\rangle$ or $|\sigma_z= -1\rangle$  \cite{Pan2019}. For convenience, we do not consider the effect of this term, which does not affect our conclusions. At Dirac point $k_x=\pi$, the term $-2\lambda_{\text{so}}\delta k_x\tau_z\sigma_zs_z$ vanish and the zero-energy states are coupled by the superconducting pairing, leading to the superconducting gap with magnitude $\Delta_{\text{A}}$ ($\Delta_{\text{B}}$) when the zigzag edge is formed by the A (B) atom. For the bearded edge, the boundary Dirac points are located at $k_x=0$, which are also gapped by the superconducting pairing with magnitude $\Delta_{\text{A}}$ ($\Delta_{\text{B}}$) when the bearded edge is formed by the A (B) atom. The zigzag and bearded edge states spectra are shown in Figs.~\ref{edge spectra} (c-d) and Figs.~\ref{edge spectra} (e-f), respectively.


\section{Majoaran Kramers Pairs in the square- and dodecagonal-shaped systems}
\label{appendixe}

As explained in the main text, Majoaran Kramers Pairs (MKP) can be realized at a corner shared by armchair and zigzag (bearded) edges when $(\Delta_{\text{A}}+\Delta_{\text{B}})\Delta_{\alpha}<0$ (the zigzag (bearded) edge is formed by the $\alpha$ atom).
As the angles between the armchair and zigzag (bearded) edges can take the values of $\pi/6$ ,$\pi/2$ and $5\pi/6$, we can engineer MKP by tailoring the systems into triangular, diamond,  square, and dodecagonal shapes.
For these systems, the edge-state energy gaps of armchair edges have the same sign, which constrains possible patterns of MKP realized through edge cleavage. In the following, we take the square- and dodecagonal-shaped systems as examples.

For a square-shaped system,  MKP can be realized at two adjacent corners [Fig.~\ref{sq}(a)] or all four corners [Fig.~\ref{sq}(b)]. To characterize the two patterns of MKP described by the $\bm{\chi}_{4}$ vector [Fig.~\ref{sq}], we calculate the spin Bott index ${\mathcal{N}}_{4}^{(1,2,3)}$, which are generated by choosing
\begin{eqnarray}
{f}_4^{(1)}=2xy/L^2,{f}_4^{(2)}=x/L,{f}_4^{(3)}=y/L,
\end{eqnarray}
where $(x,y)$ denotes the  coordinate of lattice sites and $L$ is the side length. In Figs.~\ref{sq} (a) and \ref{sq}(b), we present the numerical values of ${\mathcal{N}}_{4}$ which fully capture the different patterns of MKP described by $\bm{\chi}_4=\mathcal{M}^{-1}\cdot\left({\mathcal{N}}_4^{(1)},{\mathcal{N}}_4^{(2)},{\mathcal{N}}_4^{(3)},0\right)^{\mathrm{T}}$ with 
\begin{equation}
\mathcal{M}=\frac{1}{2}\begin{pmatrix}
    1 & -1 & 1 &-1 \\
    1 & -1 & -1 &1 \\
    1 & 1 & -1 &-1 \\
    1 & 1 & 1 &1
\end{pmatrix}.
\end{equation}

For the dodecagon-shaped system, MKP can be realized at $q$ corners by edge cleavage with $q=2,4,6,8,10,12$, as exemplified in Figs.~\ref{12m}(a)-~\ref{12m}(f).
 According to our theory, eleven spin Bott indices ${\mathcal{N}}_{12}^{(1,\cdots,11)}$ are needed to fully characterize the different patterns of MKP. The invariants ${\mathcal{N}}_{12}^{(1,\cdots,11)}$ can be generated by choosing \cite{Jiazheng2024}
 \begin{widetext}
\begin{eqnarray}
&&{f}_{12}^{(1)}=(-6 x^5y + 20x^3 y^3 - 6 xy^5)/2L^6,\nonumber\\
&&{f}_{12}^{(2)}=-(2/9) (\sqrt{3} x^6 - 21 x^5 y - 15 \sqrt{3} x^4 y^2 +  6 x^3 y^3 + 15 \sqrt{3} x^2 y^4 - 21 x y^5 - \sqrt{3} y^6)/2L^6,\nonumber\\
&&{f}_{12}^{(3)}=1/3 (3 x^6 -  (8 + 4\sqrt{3}) x^5 y - (15 + 16\sqrt{3}) x^4 y^2 +16 x^3 y^3 + (-15 + 16 \sqrt{3}) x^2 y^4+  (-8 + 4\sqrt{3}) x y^5 + 3y^6)/2L^6,\nonumber\\
&&{f}_{12}^{(4)}=-(2/9) (2 \sqrt{3} x^6 - 3 x^5 y + 18 \sqrt{3} x^4 y^2 + 42 x^3 y^3 - 18 \sqrt{3} x^2 y^4 - 3 x y^5 - 2 \sqrt{3} y^6)/2L^6,\nonumber\\
&&{f}_{12}^{(5)}=((9 + \sqrt{3}) x^6 + (24 + 12\sqrt{3}) x^5 y -(45 + 63 \sqrt{3}) x^4 y^2 - 48 x^3 y^3 -  (45 - 63 \sqrt{3}) x^2 y^4 + (24 -12 \sqrt{3}) x y^5 + (9 - \sqrt{3}) y^6)/18L^6,\nonumber\\
&&{f}_{12}^{(6)}=-\sqrt{2} y ((-1 + 6 \sqrt{3}) x^4 +  2 (7 - 4 \sqrt{3}) x^2 y^2 + (-1 + 2 \sqrt{3}) y^4)/6L^5,\nonumber\\
&&{f}_{12}^{(7)}=\sqrt{2} ((1 + \sqrt{3}) x^5 + 3 (-2 + \sqrt{3}) x^4 y -4 (-1 + \sqrt{3}) x^3 y^2 + 2 (9 - 2 \sqrt{3}) x^2 y^3 + (-5 + 3 \sqrt{3}) x y^4 + \sqrt{3} y^5)/6L^5,\nonumber\\
&&{f}_{12}^{(8)}=-\sqrt{2} (x^5 - 6 \sqrt{3} x^4 y - 14 x^3 y^2 + 8 \sqrt{3} x^2 y^3 + x y^4 - 2 \sqrt{3} y^5)/6L^5,\nonumber\\
&&{f}_{12}^{(9)}=1/3 \sqrt{2} ((-1 + \sqrt{3}) x^5 + (4 + 3 \sqrt{3}) x^4 y - 4 (1 + \sqrt{3}) x^3 y^2 +  2 (5 - 2 \sqrt{3}) x^2 y^3 + (5 + 3 \sqrt{3}) x y^4 + (-2 + \sqrt{3}) y^5)/2L^5,\nonumber\\
&&{f}_{12}^{(10)}=\sqrt{2} ((-1 + \sqrt{3}) x^5 - (4 + 3 \sqrt{3}) x^4 y - 4 (1 + \sqrt{3}) x^3 y^2 + 2 (-5 + 2 \sqrt{3}) x^2 y^3 + (5 + 3 \sqrt{3}) xy^4 - (-2 +  \sqrt{3}) y^5)/6L^5,\nonumber\\
&&{f}_{12}^{(11)}=-(1/3) \sqrt{2} (x - \sqrt{3} y) (x^4 + 7 \sqrt{3} x^3 y + 7 x^2 y^2 - \sqrt{3} x y^3 - 2 y^4)/2L^5,
\end{eqnarray}
 \end{widetext}
where $L$ is the length between the origin and a corner.
 In Figs.~\ref{12m}(a)-\ref{12m}(f), we present the numerical values of ${\mathcal{N}}_{12}$ which fully characterize the different patterns of MKP described by $\bm{\chi}_{12}$ and we have $\bm{\chi}_{12}=\mathcal{M}^{-1}\cdot\left({\mathcal{N}}_{12}^{(1)},\cdots,{\mathcal{N}}_{12}^{(11)},0\right)^{\mathrm{T}}$ with
 \begin{widetext}
\begin{equation}
\mathcal{M}=\frac{1}{2}\left(
\begin{array}{cccccccccccc}
 -1 & 1 & -1 & 1 & -1 & 1 & -1 & 1 & -1 & 1 & -1 & 1 \\
 1 & 1 & 1 & -1 & -1 & -1 & 1 & 1 & 1 & -1 & -1 & -1 \\
 -1 & -1 & 1 & 1 & -1 & 1 & -1 & -1 & 1 & 1 & -1 & 1 \\
 -1 & -1 & 1 & 1 & 1 & -1 & -1 & -1 & 1 & 1 & 1 & -1 \\
 1 & -1 & 1 & 1 & -1 & -1 & 1 & -1 & 1 & 1 & -1 & -1 \\
 -1 & -1 & -1 & -1 & -1 & -1 & 1 & 1 & 1 & 1 & 1 & 1 \\
 1 & 1 & 1 & 1 & 1 & -1 & -1 & -1 & -1 & -1 & -1 & 1 \\
 1 & 1 & 1 & 1 & -1 & 1 & -1 & -1 & -1 & -1 & 1 & -1 \\
 1 & 1 & 1 & -1 & 1 & 1 & -1 & -1 & -1 & 1 & -1 & -1 \\
 -1 & -1 & 1 & -1 & -1 & -1 & 1 & 1 & -1 & 1 & 1 & 1 \\
 -1 & 1 & -1 & -1 & -1 & -1 & 1 & -1 & 1 & 1 & 1 & 1 \\
 1 & 1 & 1 & 1 & 1 & 1 & 1 & 1 & 1 & 1 & 1 & 1 \\
\end{array}
\right).
\end{equation}
\end{widetext}

\begin{table}[htb]
\centering
\setlength\tabcolsep{5.3pt}
\renewcommand{\arraystretch}{1.5}
\caption{ The inversion eigenvalues $\xi_{\bm{k},n}$ of all the four occupied
BdG bands $n = 1, 2,3,4$, at
$\bm{k} =\bm{\Gamma},\bm{M_1},\bm{M_2},\bm{M_3}$, with $\bm{\Gamma}=(0, 0), \bm{M_1}=(0, \frac{2\pi}{\sqrt{3}}), \bm{M_2}=(\pi,\frac{\pi}{\sqrt{3}})$ and $\bm{M_3}=(-\pi,\frac{\pi}{\sqrt{3}})$.}
\begin{tabular}{|c|c|c|c|c|}
\hline
$n$&1&  2& 3 & 4 \\
\hline
$\bm{\Gamma}$&$+$&$+$&$+$&$+$\\
\hline
$\bm{M_1}$&$ -$&$ -$& $ -$&$ -$ \\
\hline
$\bm{M_2}$&$+$&$+$&$+$&$+$\\\hline
$\bm{M_3}$&$+$&$+$&$+$&$+$\\
\hline
\end{tabular}
\label{tab2}
\end{table}

\section{\texorpdfstring{$Z_4$}{} symmetry indicator}
\label{appendixg}
When $\Delta_{\text{A}}=-\Delta_{\text{B}}=\Delta$, the superconducting pairing potential can be represented by $\hat{\Delta}=\Delta\tau_x\sigma_zs_0$ under the basis of $\Psi_{\bm k}$. The system is an odd parity superconductor with the inversion symmetry defined as $\hat{I}=\tau_z \sigma_x$, which satisfies $\hat{I}{H}(\bm k)\hat{I}^{-1}={H}(-\bm k)$. At the inversion symmetry invariant momenta, the occupied Bloch states are the eigenstates of  $\hat{I}$. 
A $Z_4$ symmetry indicator can be defined as follows \cite{Hsu2020}
\begin{eqnarray}
\kappa=\frac{1}{4}\sum_{\bm k \in \mathrm{TRIM}} \sum_{n} \xi_{\bm k,n} \,\,\,\, \text{mod} \,\,4,
\label{seq7}
\end{eqnarray}
where $\xi_{\bm k,n}$ is the parity eigenvalue of the occupied $n$th BdG band at time-reversal invariant momenta $\bm k$ and $\kappa=2$ characterizes the inversion-protected MKP \cite{Hsu2020}. In Table \ref{tab2}, we present the parity eigenvalue of each occupied BdG band at the four time-reversal invariant momenta and we have $\kappa=2$ by applying Eq.~\eqref{seq7} to ${H}(\bm k)$.

\bibliography{reference1,ref-exact}

\end{document}